\documentclass[a4paper]{article}
\usepackage[margin=2.7cm]{geometry}
\usepackage[utf8]{inputenc}
\usepackage[english]{babel}
\usepackage{amsmath,amsfonts,amssymb}
\usepackage{amsthm}
\usepackage{bbold}
\usepackage[hyperindex,breaklinks=true]{hyperref}
\usepackage{textcomp}
\usepackage{mathrsfs}
\usepackage{graphicx,color}
\usepackage{todonotes}
\usepackage{mathtools}
\usepackage{caption,subcaption}
\usepackage{tikz}
\usepackage{pgfplots}
\usepackage{cite}
\usepackage{enumerate}
\usepackage{pifont}
\usepackage{hhline}
\usepackage{comment}
\usepackage{bm}

\DeclareFontFamily{OT1}{pzc}{}
\DeclareFontShape{OT1}{pzc}{m}{it}{<-> s * [1.10] pzcmi7t}{}
\DeclareMathAlphabet{\mathpzc}{OT1}{pzc}{m}{it}

\usepackage{scalerel}
\usetikzlibrary{svg.path}
\definecolor{orcidlogocol}{HTML}{A6CE39}
\tikzset{
	orcidlogo/.pic={
		\fill[orcidlogocol] svg{M256,128c0,70.7-57.3,128-128,128C57.3,256,0,198.7,0,128C0,57.3,57.3,0,128,0C198.7,0,256,57.3,256,128z};
		\fill[white] svg{M86.3,186.2H70.9V79.1h15.4v48.4V186.2z}
		svg{M108.9,79.1h41.6c39.6,0,57,28.3,57,53.6c0,27.5-21.5,53.6-56.8,53.6h-41.8V79.1z M124.3,172.4h24.5c34.9,0,42.9-26.5,42.9-39.7c0-21.5-13.7-39.7-43.7-39.7h-23.7V172.4z}
		svg{M88.7,56.8c0,5.5-4.5,10.1-10.1,10.1c-5.6,0-10.1-4.6-10.1-10.1c0-5.6,4.5-10.1,10.1-10.1C84.2,46.7,88.7,51.3,88.7,56.8z};
	}
}
\newcommand{\orcidlink}[1]{\href{https://orcid.org/#1}{\mbox{\scalerel*{
				\begin{tikzpicture}[yscale=-1,transform shape]
					\pic{orcidlogo};
				\end{tikzpicture}
			}{X}}}}

\newtheorem{theorem}{Theorem}
\newtheorem{proposition}[theorem]{Proposition}
\newtheorem{fact}[theorem]{Fact}
\newtheorem{lemma}[theorem]{Lemma}
\newtheorem{corollary}[theorem]{Corollary}

\newtheorem{definition}[theorem]{Definition}
\theoremstyle{definition}
\newtheorem{example}[theorem]{Example}

\numberwithin{equation}{section}

\newcommand{\overhat}[1]{\expandafter\hat#1}

\newcommand{\rmi}{\mathrm{i}}
\newcommand{\rme}{\mathrm{e}}
\newcommand{\rmd}{\mathrm{d}}
\newcommand{\tr}{\operatorname{tr}}
\newcommand{\ket}[1]{|{#1}\rangle}
\newcommand{\bra}[1]{\langle{#1}|}
\newcommand{\Span}{\operatorname{span}}

\definecolor{dgreen}{rgb}{0,0.5,0}

\definecolor{delete}{cmyk}{0.5,0,0,0}

\newcommand{\yes}{\color{dgreen}\ding{51}\color{black}}
\newcommand{\no}{\color{red}\ding{55}\color{black}}

\makeatletter
\def\@makefnmark{
	\leavevmode
	\raise.9ex\hbox{\fontsize\sf@size\z@\normalfont\tiny\@thefnmark}}
\makeatother

\begin{document}
\title{Bath Dynamical Decoupling with a Quantum Channel}
\author{Alexander Hahn\orcidlink{0000-0002-4152-9854}\thanks{School of Mathematical and Physical Sciences, Macquarie University, 2109 NSW, Australia}\textsuperscript{\, ,}\footnote{Corresponding author, email: \href{mailto:alexander.hahn@hdr.mq.edu.au}{alexander.hahn@hdr.mq.edu.au}} \and Kazuya Yuasa\orcidlink{0000-0001-5314-2780}\thanks{Department of Physics, Waseda University, Tokyo 169-8555, Japan}\and Daniel Burgarth\orcidlink{0000-0003-4063-1264}\footnote{Department Physik, Friedrich-Alexander-Universit\"at Erlangen-N\"urnberg, Staudtstra\ss e 7, 91058 Erlangen, Germany}\textsuperscript{\, ,}\footnotemark[1]}
\date{}
\maketitle

\begin{abstract}
	Bang-bang dynamical decoupling protects an open quantum system from decoherence due to its interaction with the surrounding bath/environment.
	In its standard form, this is achieved by strongly kicking the system with cycles of unitary operations, which average out the interaction Hamiltonian.
	In this paper, we generalize the notion of dynamical decoupling to repeated kicks with a quantum channel, which is applied to the bath.
	We derive necessary and sufficient conditions on the employed quantum channel and find that bath dynamical decoupling works if and only if the kick is ergodic.
	Furthermore, we study in which circumstances CPTP kicks on a mono-partite quantum system induce quantum Zeno dynamics with its Hamiltonian cancelled out.
	This does not require the ergodicity of the kicks, and the absence of decoherence-free subsystems is both necessary and sufficient.
	While the standard unitary dynamical decoupling is essentially the same as the quantum Zeno dynamics, our investigation implies that this is not true any more in the case of CPTP kicks.
	To derive our results, we prove some spectral properties of ergodic quantum channels, that might be of independent interest.
	Our approach establishes an enhanced and unified mathematical understanding of several recent experimental demonstrations and might form the basis of new dynamical decoupling schemes that harness environmental noise degrees of freedom.
\end{abstract}

\section{Introduction}
\label{sec:Introduction}
One of the main challenges in the development of quantum technology is the intrinsic coupling of quantum systems to their environments that cause decoherence~\cite{Preskill2018}\@.
A promising approach to overcome this hurdle is the technique of \emph{dynamical decoupling} (DD)~\cite{Viola1998, Ban1998, Viola1999, Viola1999a, Viola2002, Viola2003, Cappellaro2006, Uhrig2007}, which has been developed in the late 1990s and early 2000s employing old ideas from nuclear magnetic resonance (NMR)~\cite{Hahn1950, Carr1954, Meiboom1958, Waugh1968, Haeberlen1968}\@.
DD is a robust open-loop strategy based on fast and strong (bang-bang) controls to average out unwanted interactions. 
As such, it can be used, for example, to improve the quality of quantum computations~\cite{Viola1999a, Yang2010, West2010, Souza2011} or quantum memory~\cite{Biercuk2009, Souza2011, Peng2011}\@.
Nowadays, there exists an entire zoo of different DD schemes: from simple spin echos~\cite{Carr1954, Meiboom1958}, over time-optimized sequences~\cite{Uhrig2007}, to group-based~\cite{Viola1999, Viola1999a,Hahn2024}, embedded~\cite{Khodjasteh2005, Khodjasteh2007}, or randomized~\cite{Viola2005, Santos2006} methods.
See for instance Refs.~\cite{Viola2006, Santos2008} for a comparison of some of the available schemes.
Furthermore, a discussion of DD with finite-strength pulses beyond bang-bang controls can be found, e.g.\ in Ref.~\cite{Viola2003}\@.
Nevertheless, the underlying idea of all DD procedures is the same: one intersperses the dynamics of a quantum system with cycles of unitary operations.
If these kicks are applied fast enough, their induced rotations effectively average out undesired system-bath interactions.
In real physical systems, this can be implemented e.g.~through laser or microwave pulses and the timescale of the kicks is determined by the system-bath coupling strength.

Usually, the premise in DD is that we can only control the system itself and do not have access to the environmental degrees of freedom.
However, there are physical scenarios where it makes sense to relax this restriction.
For example, a recent theoretical study~\cite{burgelman_quantum_2022} shows that coupling the bath to a larger environment can increase the coherence time of the system through a DD effect induced by the interaction between the bath and the larger environment.
Furthermore, there have been several experimental demonstrations where controlling the bath leads to an enhanced lifetime of the coherence of the system.
See for instance Refs.~\cite{Hanson2008, Lange2012, Knowles2013, Dong2019, Joos2022, Uysal2023}\@.
Other experimental works actually utilized the effect of the system-bath interaction for system manipulation and control~\cite{Nakajima2018, Dasari2021}\@.
Conversely, it has been experimentally observed that acting on the bath with the ``wrong'' operations can also completely destroy the system coherence~\cite{Takahashi2008}\@.
Such a response is similar to the anti-Zeno effect, where the system decay rate is enhanced through repeated measurements~\cite{Facchi2001a} with an unfavorable time modulation~\cite{Chaudhry2016}\@.
This raises the question of which operations shall be applied to the bath to achieve DD and, in turn, a longer system coherence.
However, so far there is no unified framework that imposes general decoupling conditions on the bath operations.
This paper aims to solve this problem by considering the most general class of bath operations for DD, namely quantum channels.
In this setting of controlling the bath, it makes sense to relax the constraint of unitary kicks.
This is because we do not need to reverse the applied operations at the end of the decoupling sequence.
We only care about the final \emph{system} state, which is not affected by the bath pulses.
By this approach, we do not only inaugurate a unified systematic approach to mathematically describing existing experiments.
Also, we pave the way for the development of novel DD schemes that explicitly harness quantum noise and classical uncertainty.

The paper is structured as follows.
In Sec.~\ref{sec:preliminaries}, we give a short mathematical introduction.
This covers establishing our notation in Sec.~\ref{sec:notation} as well as recalling some standard results on unitary system DD in Sec.~\ref{sec:system_DD}\@.
Afterwards, we discuss the spectral properties of quantum channels in Sec.~\ref{sec:spectral_properties}\@.
We divide this excursus into two parts.
First, we summarize some known general results in Sec.~\ref{sec:properties_CPTP}\@.
Second, we study ergodic quantum channels in Sec.~\ref{sec:ergodic_channels} and characterize their spectral structure.
These results are used to prove our first main result in Sec.~\ref{sec:bath_DD}\@.
Here, we establish the concept of bath dynamical decoupling and show that it works if and only if the applied quantum channel is ergodic.
This is followed by a discussion in Sec.~\ref{sec:dynamical_freezing} on suppressing a mono-partite Zeno Hamiltonian through repeated applications of a quantum channel.
We find that the weaker condition of the absence of decoherence-free subsystem is both necessary and sufficient in this case.
Our investigation is complemented by case studies with some selected examples in Sec.~\ref{sec:examples}\@.
Finally, we conclude in Sec.~\ref{sec:conclusion}\@.

\section{Mathematical Preliminaries}
\label{sec:preliminaries}
In this section, we introduce some useful prerequisites.
This includes essential notation used throughout this paper and a short review of standard DD with unitary kicks.
The latter will also help to compare our scheme with existing methods.

\subsection{Notation}
\label{sec:notation}
Let us first fix some basic notation.
In this work, we consider a quantum system on a finite-dimensional Hilbert space $\mathscr{H}$ with $\dim\mathscr{H}=d<\infty$.
The space of the operators that act on the Hilbert space $\mathscr{H}$ is denoted by $\mathcal{B}(\mathscr{H})$.
As a subspace, it contains the space of density operators $\mathcal{T}(\mathscr{H})=\{\rho\in\mathcal{B}(\mathscr{H}):\rho\geq0,\ \tr(\rho)=1\}$.
Physical operations on the space $\mathcal{T}(\mathscr{H})$ that map density operators to density operators are completely positive trace-preserving (CPTP) linear maps $\mathcal{E}$, also known as ``quantum channels''~\cite{Alicki2007, Nielsen2010, Wolf2012, Chruscinski2017, Watrous2018, Chruscinski2022}\@.
They satisfy $\mathcal{E}(\mathcal{T}(\mathscr{H}))=\mathcal{T}(\mathscr{H})$ and are stable under tensoring arbitrary finite-dimensional ancilla systems $\mathscr{H}'$, i.e.~$(\mathcal{E}\otimes\mathcal{I}_{\mathscr{H}'})(\mathcal{T}(\mathscr{H}\otimes\mathscr{H}'))=\mathcal{T}(\mathscr{H}\otimes\mathscr{H}')$.
Here, $\mathcal{I}_{\mathscr{H}'}$ is the identity map on $\mathcal{B}(\mathscr{H}')$, i.e.~$\mathcal{I}_{\mathscr{H}'}(X)=X$ for all $X\in\mathcal{B}(\mathscr{H}')$.
In fact, the stability under tensoring ancilla systems is a physically necessary constraint ensuring that a density operator is again mapped to a valid density operator under $\mathcal{E}$ even in the presence of a spectator.
The characterizations of the quantum channel are summarized as follows.
\begin{definition}[Quantum channel]\label{def:quantum_channel}
	A linear map $\mathcal{E}:\mathcal{B}(\mathscr{H})\rightarrow\mathcal{B}(\mathscr{H})$ is called ``quantum channel'' or ``CPTP map,'' if it satisfies
	\begin{enumerate}[(i)]
		\item positivity: $\mathcal{E}(X)\geq 0$ for all $X\in\mathcal{B}(\mathscr{H})$ satisfying $X\geq 0$;
		\item complete positivity: if $X\in \mathcal{T}(\mathscr{H}\otimes\mathscr{H}')$ satisfies $X\geq0$, then $(\mathcal{E}\otimes \mathcal{I}_{\mathscr{H}'})(X)\geq0$ for all finite-dimensional ancilla systems $\mathscr{H}'$;
		\item\label{it:trace_preservation} trace-preservation: $\tr(\mathcal{E}(X))=\tr(X)$.
	\end{enumerate}
	Linear maps $\mathcal{E}:\mathcal{B}(\mathscr{H})\rightarrow\mathcal{B}(\mathscr{H})$ for which condition~(iii) is relaxed to $\tr(\mathcal{E}(X))\leq\tr(X)$ for all $X\in\mathcal{B}(\mathscr{H})$ satisfying $X\ge0$ are called ``quantum operations.''
\end{definition}

Quantum operations describe various general physical instruments that can be applied to quantum states.
A quantum channel $\mathcal{E}:\mathcal{B}(\mathscr{H})\rightarrow\mathcal{B}(\mathscr{H})$ is uniquely specified by its \emph{Choi-Jamio{\l}kowski state} $\Lambda(\mathcal{E})\in\mathcal{B}(\mathscr{H}\otimes\mathscr{H})$ via~\cite{Wolf2012, Chruscinski2017, Watrous2018, Chruscinski2022}
\begin{equation}
	\Lambda(\mathcal{E})=(\mathcal{E}\otimes\mathcal{I}_\mathscr{H})(\ket{\Omega}\bra{\Omega}),\label{eq:Choi-Jamiolkowski}
\end{equation}
where $\ket{\Omega}=\frac{1}{\sqrt{d}}\sum_{j=1}^d \ket{jj}$ is the maximally entangled state.
In the case of a closed system, the time-evolution of the system is described by a \emph{unitary} quantum channel $\mathcal{U}:\mathcal{B}(\mathscr{H})\rightarrow\mathcal{B}(\mathscr{H})$, which satisfies $\mathcal{U}^\dagger\mathcal{U}=\mathcal{U}\mathcal{U}^\dagger=\mathcal{I}_\mathscr{H}$, where the dagger $\dagger$ denotes the adjoint.
The Choi-Jamio{\l}kowski state $\Lambda(\mathcal{U})$ of a unitary quantum channel is pure~\cite{Hahn2022}\@.
The dynamics of a closed quantum system is then governed by a Hamiltonian $H\in\mathcal{B}(\mathscr{H})$, which is a Hermitian operator, $H=H^\dagger$, acting on the Hilbert space $\mathscr{H}$.
We denote its adjoint representation with a calligraphic letter as $\mathcal{H}=[H,{}\bullet{}]$, and will also use the terminology ``Hamiltonian'' for $\mathcal{H}$.
This makes sense because $\mathcal{H}$ is Hermitian, $\mathcal{H}=\mathcal{H}^\dag$, and is the generator of a one-parameter semi-group $\mathcal{U}_t=\rme^{-\rmi t\mathcal{H}}$, $t\in\mathbb{R}$, consisting of unitary quantum channels acting on $\mathcal{B}(\mathscr{H})$.
These unitary quantum channels describe the dynamics on $\mathcal{T}(\mathscr{H})$ under $\mathcal{H}$ in the same way as the unitaries $U_t=\rme^{-\rmi tH}$ describing the dynamics on $\mathscr{H}$ under $H$---a fact that is based on Stone's theorem~\cite[Sec.~VIII.4]{Reed1980}\@.
Again, we remark that we deal with finite-dimensional quantum systems in this work.
However, more information about the unitary closed system dynamics in the density operator picture for the general (possibly infinite-dimensional) case can be found in Ref.~\cite{Lonigro2024}\@.

In the physical setting of DD, we are given some finite-dimensional Hilbert space consisting of a system part $\mathscr{H}_1$ and a bath/environmental part $\mathscr{H}_2$.
Thus, the total Hilbert space $\mathscr{H}=\mathscr{H}_1\otimes\mathscr{H}_2$ is endowed with a bipartite structure.
Here, we will assume that $d_1=\dim\mathscr{H}_1\geq 2$ as well as $d_2=\dim\mathscr{H}_2\geq 2$.
Notice that the subdivision into $\mathscr{H}_1$ and $\mathscr{H}_2$ is arbitrary, and we could in principle choose it however we want to.
Usually, there is a physical constraint that leads to a natural subdivision.
Without loss of generality, we can write an arbitrary Hamiltonian $H$ on $\mathscr{H}=\mathscr{H}_1\otimes\mathscr{H}_2$ in the form
\begin{equation}
	H=H_1\otimes \mathbb{1}_2 + \mathbb{1}_1\otimes H_2 + \sum_{i=1}^D h_1^{(i)}\otimes h_2^{(i)},
	\label{eq:Hamiltonian_decomposition}
\end{equation} 
by exploiting an operator Schmidt decomposition~\cite{Tyson2003}\@.
Here, $\mathbb{1}_{1(2)}$ is the identity matrix on $\mathscr{H}_{1(2)}$ and $D\leq \min\{\dim(\mathscr{H}_1)^2,\dim(\mathscr{H}_2)^2\}$.
Again, without loss of generality, we can always shift the energy, so that $\tr(H_1)=\tr(H_2)=\tr(h_1^{(i)})=\tr(h_2^{(i)})=0$, for all $i$.
This is because the traces can be collected into a single term, which is proportional to the identity $\mathbb{1}_1\otimes\mathbb{1}_2$ and only results in an irrelevant global phase in the evolution under $H$.
The action of the adjoint representation $\mathcal{H}=[H,{}\bullet{}]$ of $H$ on a product operator $A=A_1\otimes A_2$ can then be written as
\begin{equation}
\mathcal{H}(A)=[H_1,A_1]\otimes A_2 + A_1\otimes [H_2,A_2] + \sum_{i}\left(
[h_1^{(i)},A_1]\otimes h_2^{(i)}A_2+A_1 h_1^{(i)}\otimes[h_2^{(i)},A_2]
\right).
\label{eq:Hamiltonian_action}
\end{equation}
The free evolution under the Hamiltonian in Eq.~\eqref{eq:Hamiltonian_decomposition} or~\eqref{eq:Hamiltonian_action} couples the system $\mathscr{H}_1$ to the bath $\mathscr{H}_2$.
This coupling causes decoherence in the reduced dynamics of system $\mathscr{H}_1$ over time.
The goal of DD is to suppress this decoherence by dynamically cancelling the interaction part of the Hamiltonian.
In the next subsection, we introduce how this is usually done in the setting of the standard system DD\@.

\subsection{Recap of the Standard Unitary System Dynamical Decoupling}
\label{sec:system_DD}
For completeness and to compare our bath DD framework with existing methods, we here recapitulate the standard system DD\@.
This method works by frequently kicking the system $\mathscr{H}_1$ with cycles of unitary rotations~\cite{Viola1999, Viola1999a, Viola2003, Khodjasteh2005}\@.
To mathematically describe this technique, consider a set $\mathscr{V}=\{V_i\}_{i=1}^M$ of unitary operators $V_i\in\mathcal{B}(\mathscr{H})$.
This set of size $\vert\mathscr{V}\vert=M$ is called the \emph{decoupling set}\@.
The operations with $V_i$ are called \emph{decoupling operations} or \emph{unitary kicks} and have the form $V_i=v_i\otimes\mathbb{1}_2$, where $v_i\in\mathcal{B}(\mathscr{H}_1)$ are unitary operators on $\mathscr{H}_1$.
This particular form of $V_i$ incorporates the usual premise that we can only control the system but not the bath.
Furthermore, we require
\begin{equation}
	\frac{1}{M}\sum_{i=1}^M v_iXv_i^\dagger=\frac{1}{d_1}\tr(X)\mathbb{1}_1,\quad \forall X\in\mathcal{B}(\mathscr{H}_1),\label{eq:decoupling_condition}
\end{equation}
which is sometimes referred to as \emph{decoupling condition}\@.
Condition~\eqref{eq:decoupling_condition} requires $\mathscr{V}$ to be a \emph{unitary 1-design}~\cite{Roy2009, Mele2023}\@.
Notice that the channel constructed in Eq.~\eqref{eq:decoupling_condition} by averaging over the set $\mathscr{V}$ is the \emph{completely depolarizing channel} $\mathcal{P}^{\mathbb{1}/d_1}=\tr({}\bullet{})\mathbb{1}_1/d_1$ acting on $\mathcal{B}(\mathscr{H}_1)$.
It is an orthogonal projection satisfying $(\mathcal{P}^{\mathbb{1}/d_1})^2=(\mathcal{P}^{\mathbb{1}/d_1})^\dagger=\mathcal{P}^{\mathbb{1}/d_1}$.
In the language of designs, this operator is referred to as \emph{twirl} over the unitary 1-design $\mathscr{V}$, see e.g.~Ref.~\cite{Conrad2021}\@.

In standard system DD, the unitary dynamics of a quantum system on $\mathscr{H}$ is interspersed cyclicly by the decoupling operations $\{V_i\}$.
This results in a dynamics given by
\begin{equation}
	U_{t,n}^\mathrm{DD}=\left(V_M \rme^{-\rmi\frac{t}{nM}H}\cdots V_2 \rme^{-\rmi\frac{t}{nM}H}V_1 \rme^{-\rmi\frac{t}{nM}H}\right)^n.\label{eq:unitary_DD}
\end{equation}
It has been shown in Ref.~\cite{Facchi2004} that this evolution leads to a variant of the \emph{quantum Zeno dynamics}~\cite{Facchi2002, Facchi2008, burgarth_quantum_2020, Burgarth2022} in the decoupling limit $n\rightarrow\infty$,
\begin{equation}
	\lim_{n\rightarrow\infty}U_{t,n}^\mathrm{DD}=(V_M\cdots V_2 V_1)^n \rme^{-\rmi t H_\mathrm{Z}} + \mathcal{O}(1/n),
	\label{eq:system_DD_Zeno}
\end{equation}
where the Zeno Hamiltonian $H_\mathrm{Z}$ is given by
\begin{equation}
	H_\mathrm{Z}=(\mathcal{P}^{\mathbb{1}/d_1}\otimes \mathcal{I}_2)(H),
	\label{eq:Zeno_Hamiltonian_system_DD}
\end{equation}
with $\mathcal{I}_{2(1)}$ the identity map on $\mathcal{B}(\mathscr{H}_{2(1)})$.
Equation~\eqref{eq:Zeno_Hamiltonian_system_DD} can be rewritten as
\begin{equation}
	H_\mathrm{Z} = \mathbb{1}_1 \otimes H_2,
	\label{eq:Dec_Zeno_Hamiltonian_system_DD}
\end{equation}
by using Eq.~\eqref{eq:Hamiltonian_decomposition} and $\tr(H_1)=\tr(h_1^{(i)})=0$, for all $i$.
Thus, the evolution under $\lim_{n\rightarrow\infty} U_{t,n}^\mathrm{DD}$ is indeed decoupled [up to an error that scales as $\mathcal{O}(1/n)$].
If we reverse the overall unitary rotation $(V_M\cdots V_2 V_1)^n$ induced by the full decoupling sequence at the end of the evolution, we recover the initial system state.
Exactly the same considerations can even be made for \emph{random dynamical decoupling}~\cite{Viola2005,Santos2006,Viola2006,Santos2008}, where Eq.~\eqref{eq:unitary_DD} is replaced by a product $V_n \rme^{-\rmi\frac{t}{n}H}\cdots V_2 \rme^{-\rmi\frac{t}{n}H}V_1 \rme^{-\rmi\frac{t}{n}H}$, in which the decoupling operations $V_i$ are sampled from an independent and identically distributed (i.i.d.) random variable at each time step $i$.
As for the deterministic case~\eqref{eq:unitary_DD}, random DD is a manifestation of the quantum Zeno dynamics, where the effective Zeno Hamiltonian is given by Eqs.~\eqref{eq:Zeno_Hamiltonian_system_DD}--\eqref{eq:Dec_Zeno_Hamiltonian_system_DD}\@.
See Refs.~\cite{Hillier2015,Hahn2022} for details.

The unitarity of the kicks is essential in these schemes since we want to retain the information about the system state.
However, imagine that we have control over the bath degrees of freedom and still only care about the system state at the end of the evolution.
In this case, a decoupling scheme completely destroying the bath state would still be reasonable.
In Ref.~\cite{Hahn2022}, we revealed the existence of a symmetry between the system and the bath that allows us to focus on the bath dynamics only: if the bath evolves unitarily, the system will do so as well.
This fact, which we call \emph{equitability of system and bath}, naturally leads to the concept of bath DD, where the system is protected through decoupling cycles acting on the bath.
Since the final bath state does not matter, we can even apply CPTP kicks to the bath, instead of the unitary kicks.
From Eq.~\eqref{eq:Zeno_Hamiltonian_system_DD}, we can deduce that a unitary decoupling cycle effectively acts as the completely depolarizing channel.
This gives us the first candidate of a quantum channel for bath DD\@.
Indeed, from Ref.~\cite[Corollary~2]{burgarth_quantum_2020}, we can infer that the bath DD scheme
\begin{equation}
	\lim_{n\rightarrow\infty}\left((\mathcal{I}_1\otimes\mathcal{P}^{\mathbb{1}/d_2})\rme^{-\rmi\frac{t}{n}\mathcal{H}}\right)^n = (\mathcal{I}_1\otimes\mathcal{P}^{\mathbb{1}/d_2})\rme^{-\rmi t (\mathcal{H}_1\otimes \mathcal{I}_2)}
\end{equation}
using the bath quantum channel $\mathcal{I}_1\otimes\mathcal{P}^{\mathbb{1}/d_2}$ gives rise to a decoupled evolution.
See also Refs.~\cite{Matolcsi2003,Barankai2018,Moebus2019,Becker2021,Moebus2023,Moebus2024,Salzmann2024}\@.
Here, $\mathcal{H}_1=[H_1,{}\bullet{}]$.
The question is then whether this observation can be generalized to other quantum channels $\mathcal{E}_2$ acting on $\mathcal{B}(\mathscr{H}_2)$ via
\begin{equation}
	\lim_{n\rightarrow\infty}\left((\mathcal{I}_1\otimes\mathcal{E}_2)\rme^{-\rmi\frac{t}{n}\mathcal{H}}\right)^n.
	\label{eq:bath_DD_def}
\end{equation}
Here, we answer this question and derive both necessary and sufficient conditions on the bath quantum channel $\mathcal{E}_2$ giving rise to a decoupled evolution according to Eq.~\eqref{eq:bath_DD_def}\@.
For this purpose, we will have to introduce some results on the spectral properties of quantum channels.

\section{Spectral Properties of Quantum Channels}
\label{sec:spectral_properties}
In this section, we discuss the mathematical preliminaries on quantum channels, which we need to study bath DD\@.
Section~\ref{sec:properties_CPTP} collects some known results on quantum channels, particularly regarding their spectral properties.
These results are not new but rather stated for completeness and the reader's convenience.
It is followed by a discussion on ergodic quantum channels in Sec.~\ref{sec:ergodic_channels}\@.

\subsection{General Quantum Channels}
\label{sec:properties_CPTP}
The key players in our discussion are CPTP maps $\mathcal{E}:\mathcal{B}(\mathscr{H})\rightarrow\mathcal{B}(\mathscr{H})$, also known as quantum channels.
See Definition~\ref{def:quantum_channel} for the definition and Refs.~\cite{Alicki2007, Nielsen2010, Wolf2012, Chruscinski2017, Watrous2018, Chruscinski2022} for detailed introductions.
In this section, we recapitulate some of their properties.
In particular, we will be interested in their spectral properties, which have been discussed in depth in Ref.~\cite[Sec.~6]{Wolf2012}\@.
We summarize some facts relevant to our following discussion.
\begin{fact}[Spectral properties of quantum channels]
\label{prop:spectral_channel}
	For a quantum channel $\mathcal{E}$ on a finite-dimensional system, the following statements hold.
	\begin{enumerate}[(i)]
		\item The spectrum $\sigma(\mathcal{E})=\{\lambda_\ell\}$ of $\mathcal{E}$ is entirely contained in the closed unit disc $\overline{\mathbb{D}}=\{\lambda\in\mathbb{C}:|\lambda|\le 1\}$, and $\lambda=1$ is an eigenvalue of $\mathcal{E}$. The so-called ``peripheral spectrum'' $\sigma_\varphi(\mathcal{E})=\{\lambda_\ell:|\lambda_\ell|=1\}$ of $\mathcal{E}$ consists of the eigenvalues of $\mathcal{E}$ with magnitude $1$.
		\item $\mathcal{E}$ can be written in its spectral representation as
			\begin{equation}
				\mathcal{E}=
				\mathcal{E}_\varphi
				+\sum_{|\lambda_\ell|<1}(\lambda_\ell\mathcal{P}_\ell+\mathcal{N}_\ell),
			\end{equation}
		with
			\begin{equation}
				\mathcal{E}_\varphi = \sum_{|\lambda_\ell|=1}\lambda_\ell\mathcal{P}_\ell,
			\end{equation}
		where $\{\mathcal{P}_\ell\}$ and $\{\mathcal{N}_\ell\}$ are the spectral projections and the nilpotents of $\mathcal{E}$, respectively, and $\mathcal{E}_\varphi$ is called the ``peripheral part'' of $\mathcal{E}$.
		The spectral projections fulfill $\mathcal{P}_\ell\mathcal{P}_{\ell'}=\delta_{\ell\ell'}\mathcal{P}_\ell$ and $\sum_\ell\mathcal{P}_\ell=\mathcal{I}$, with $\mathcal{I}$ the identity map on $\mathcal{B}(\mathscr{H})$, while the nilpotents satisfy $\mathcal{N}_\ell\mathcal{P}_{\ell'}=\mathcal{P}_{\ell'}\mathcal{N}_\ell=\delta_{\ell\ell'}\mathcal{N}_\ell$ and $\mathcal{N}_\ell^{n_\ell}=0$, for some $n_\ell\in\mathbb{N}$.
		Notice that the peripheral part $\mathcal{E}_\varphi$ is free from nilpotents, and $\mathcal{E}_\varphi$ is thus diagonalizable (also called non-defective).
		Moreover, $\mathcal{E}_\varphi$ is a quantum channel, too.
		The spectral representation of $\mathcal{E}$ is unique.
		\item The projection onto the ``peripheral subspace'' corresponding to the peripheral spectrum $\sigma_\varphi(\mathcal{E})$,
			\begin{equation}
				\mathcal{P}_\varphi=\sum_{|\lambda_\ell|=1}\mathcal{P}_\ell,
			\end{equation}
		is a quantum channel. It satisfies $\mathcal{E}_\varphi=\mathcal{E}\mathcal{P}_\varphi=\mathcal{P}_\varphi\mathcal{E}$. Since $\mathcal{P}_\varphi$ does not have any nilpotents, it is diagonalizable.
		\item $\mathcal{E}_\varphi$ restricted to the range of $\mathcal{P}_\varphi$ is invertible in the sense that $\mathcal{E}_\varphi^{-1}\mathcal{E}_\varphi=\mathcal{E}_\varphi\mathcal{E}_\varphi^{-1}=\mathcal{P}_\varphi$, and the inverse is given by
			\begin{equation}
				\mathcal{E}_\varphi^{-1}=\sum_{|\lambda_\ell|=1}\lambda_\ell^{-1}\mathcal{P}_\ell,
			\end{equation}
			which is again a quantum channel.
	\end{enumerate}
\end{fact}
	\begin{proof}
		This is proved in Ref.~\cite[Sec.~6]{Wolf2012}\@.
	\end{proof}
A direct generalization of Fact~\ref{prop:spectral_channel} to quantum operations is given in Ref.~\cite[Proposition~1]{burgarth_quantum_2020}\@.
Moreover, generalizations of Fact~\ref{prop:spectral_channel} to completely positive operators on possibly infinite-dimensional Banach spaces can be found in Refs.~\cite{Groh1983, Groh1984}\@.
These results have been even further generalized in Refs.~\cite{Krengel1985, Eisner2015}\@.

Fact~\ref{prop:spectral_channel} concerns the characteristics of the spectrum as well as the spectral properties of a quantum channel, but does not discuss its eigenoperators.
We will be interested in the evolutions within the peripheral subspace of a quantum channel.
This motivates the following definition.
\begin{definition}[Recurrences and fixed points]
	The ``space of recurrences'' $\mathcal{X}(\mathcal{E})$ of a quantum channel $\mathcal{E}:\mathcal{B}(\mathscr{H})\rightarrow\mathcal{B}(\mathscr{H})$ is the complex span of all the eigenoperators of $\mathcal{E}$ belonging to its peripheral spectrum $\sigma_\varphi(\mathcal{E})$.
	That is,
	\begin{equation}
	\mathcal{X}(\mathcal{E})=\Span\{X\in\mathcal{B}(\mathscr{H}):\mathcal{E}(X)=\rme^{\rmi\phi}X,\text{ for some }\phi\in\mathbb{R}\}.
\end{equation}
$\mathcal{X}(\mathcal{E})$ contains as a subspace the ``space of fixed points'' $\mathcal{F}(\mathcal{E})$ of $\mathcal{E}$,
\begin{equation}
	\mathcal{F}(\mathcal{E})=\{X\in\mathcal{B}(\mathscr{H}):\mathcal{E}(X)=X\}.
\end{equation}
\end{definition}
\begin{fact}[Fixed-point states]
The space of fixed points $\mathcal{F}(\mathcal{E})$ of $\mathcal{E}$ is actually spanned by the fixed-point states of $\mathcal{E}$,
\begin{equation}
	\mathcal{F}(\mathcal{E})=\Span\{\rho\in\mathcal{T}(\mathscr{H}):\mathcal{E}(\rho)=\rho\}.
\end{equation}
Notice that every quantum channel has at least one fixed point by Fact~\ref{prop:spectral_channel}(i) and therefore $\dim\mathcal{F}(\mathcal{E})\ge1$.
\end{fact}
\begin{proof}
This is shown in Ref.~\cite[Corollary~6.5]{Wolf2012}.
\end{proof}

The space of fixed points and the space of recurrences have been studied in detail in Ref.~\cite[Secs.~6.4--6.5]{Wolf2012} as well as Ref.~\cite{Wolf2010}\@.
In particular, see Ref.~\cite[Theorems~6.14 and~6.16]{Wolf2012} and Ref.~\cite[Theorem~8]{Wolf2010}\@.
See also Refs.~\cite{Amato2023, Amato2023a, Amato2024, Amato2024a}\@.
We summarize the results as follows.
\begin{fact}[Structure of the space of recurrences]\label{prop:recurrences}
	For every quantum channel $\mathcal{E}:\mathcal{B}(\mathscr{H})\rightarrow\mathcal{B}(\mathscr{H})$, there exists a decomposition of the Hilbert space
	\begin{equation}
		\mathscr{H}=\mathscr{H}_0\oplus\bigoplus_{k=0}^{K-1} \mathscr{H}_{k,1}\otimes\mathscr{H}_{k,2}
	\end{equation}
	in some basis, such that the space of recurrences $\mathcal{X}(\mathcal{E})$ has a decomposition of the form
	\begin{equation}
		\mathcal{X}(\mathcal{E})=0\oplus\bigoplus_{k=0}^{K-1} \mathcal{M}_{d_k}\otimes\rho_k,
		\label{eq:recurrent_space}
	\end{equation}
	where $0$ is a zero block of size $d_0\times d_0$ with $d_0=\dim(\mathscr{H}_0)$, $\mathcal{M}_{d_k}$ is a full $d_k\times d_k$ matrix algebra acting on $\mathscr{H}_{k,1}$ with $d_k=\dim\mathscr{H}_{k,1}$, and $\rho_k\in\mathcal{T}(\mathscr{H}_{k,2})$ is a positive density matrix $\rho_k>0$ acting on $\mathscr{H}_{k,2}$.
	Furthermore, the tensor product structure inside the direct sum comes from the reduction of a full matrix algebra to a von Neumann algebra factor (see Ref.~\cite{Lindblad1999} for details).
	This means that we can find $x_k\in\mathcal{B}(\mathscr{H}_{k,1})$ such that any $X\in\mathcal{X}(\mathcal{E})$ can be written as
	\begin{equation}
		X=0\oplus\bigoplus_{k=0}^{K-1} x_k\otimes\rho_k.
		\label{eq:recurrence}
	\end{equation}
	In this basis, the action of $\mathcal{E}$ on any recurrent operator $X\in\mathcal{X}(\mathcal{E})$ reads
	\begin{equation}
		\mathcal{E}(X)=0\oplus\bigoplus_{k=0}^{K-1} U_kx_{\pi(k)}U_k^\dagger\otimes\rho_k,
		\label{eq:recurrences_map}
	\end{equation}
	where $U_k\in\mathcal{B}(\mathscr{H}_{k,1})$ are $d_k\times d_k$ unitary matrices and $\pi$ is a permutation among sub-blocks $\mathcal{M}_{d_k}$ of equal dimensions.
	That is, $\pi$ permutes the elements of $\{0,\dots, K-1\}$ within each subset of indices $k$ that share the same matrix dimension $d_k$.
	\begin{proof}
		This is proved in Ref.~\cite[Theorems~6.14 and~6.16]{Wolf2012} and Ref.~\cite[Theorem~8]{Wolf2010}\@.
	\end{proof}
\end{fact}

We now look into the structure of the spaces $\mathscr{H}_{k,1}$ in more detail and define the notion of de\-coherence-free subsystems as first introduced in Refs.~\cite{Shor1995, Plenio1997}\@.
See also Refs.~\cite{Knill2000, Zanardi2000, Kempe2001, bacon_decoherence_2001, Kribs2005, Shabani2005, Knill2006, Lidar2013, guan_structure_2018, Dash2023} for more details and applications of the decoherence-free subsystems in the context of quantum error correction.
\begin{definition}[Decoherence-free subsystem]
	Let $\mathcal{E}:\mathcal{B}(\mathscr{H})\rightarrow\mathcal{B}(\mathscr{H})$ be a finite-dimensional quantum channel. Fix the basis for $\mathscr{H}$, such that $\mathcal{X}(\mathcal{E})$ is decomposed as in Eq.~\eqref{eq:recurrent_space}\@. Then, each $\mathscr{H}_{k,1}$ with $\dim\mathscr{H}_{k,1}\ge2$ is called ``decoherence-free subsystem'' of $\mathscr{H}$ with respect to $\mathcal{E}$.
\end{definition}

In the following, we will just say that $\mathcal{E}$ has a decoherence-free subsystem.
The reason for this name is that we can factorize out the unitaries $U_k$ acting on $\mathscr{H}_{k,1}$ from $\mathcal{E}$ as $\mathcal{E}=\mathcal{U}\circ\tilde{\mathcal{E}}$.
Then, $\tilde{\mathcal{E}}$ acts trivially on $\mathscr{H}_{k,1}$ (up to possible permutations) while $\mathcal{U}$ preserves the coherence.
$\dim\mathscr{H}_{k,1}\ge2$ is demanded to allow the encoding of quantum information into $\mathscr{H}_{k,1}$.
If $\pi$ and all $U_k$ are trivial, one speaks of a \emph{noiseless} subsystem. See Refs.~\cite{Choi2006, guan_structure_2018}\@.
Note also that a decoherence-free subsystem $\mathscr{H}_{k,1}$ is called a \emph{``decoherence-free subspace''} (see e.g.~Refs.~\cite{Lidar1998,Lidar2003,guan_structure_2018}) in the literature, if $\dim\mathscr{H}_{k,2}=1$.

\begin{example}[Decoherence-free subsystem]\label{example:DFSS}
	For concreteness, we here introduce two examples of quantum channels with decoherence-free subsystems.
	\begin{enumerate}[(i)]
		\item\label{it:DFSS1} Let $\mathscr{H}=\mathbb{C}_1^2\otimes\mathbb{C}_2^2$, i.e.~a two-qubit system, where we associate the first qubit with $\mathbb{C}_1^2$ and the second qubit with $\mathbb{C}_2^2$. Define the quantum channel $\mathcal{E}^\Omega$ by its action on an arbitrary $A\in\mathcal{B}(\mathbb{C}_1^2\otimes\mathbb{C}_2^2)$ as $\mathcal{E}^\Omega(A)=\tr_2(A)\otimes\Omega$ for some fixed $\Omega\in\mathcal{T}(\mathbb{C}_2^2)$, where $\tr_2$ denotes the partial trace over $\mathbb{C}_2^2$. For example, take $\Omega=\frac{1}{2}\mathbb{1}$, so that $\mathcal{E}^{\mathbb{1}/2}(A)=\tr_2(A)\otimes\frac{1}{2}\mathbb{1}$. Then, $\mathcal{B}(\mathbb{C}_1^2)$ is a decoherence-free subsystem. This can be easily verified by looking at the decomposition of the space of recurrences of $\mathcal{E}^\Omega$ in the computational basis. For all $X\in\mathcal{X}(\mathcal{E}^\Omega)$, we have $X=x\otimes\Omega$, and hence $\mathcal{E}^\Omega(X)=x\otimes\Omega$, where $x\in\mathcal{M}_2$. Thus, $\mathcal{B}(\mathbb{C}_1^2)$ is even a noiseless subsystem.
		\item\label{it:DFSS2} The above example is rather simplistic, as it involves neither unitary action on the decoherence-free subsystem nor permutation of sub-blocks. Now, we introduce a quantum channel, which involves both. Let $\mathscr{H}=\mathbb{C}_1^2\otimes\mathbb{C}_2^2\otimes\mathbb{C}_3^2$, i.e.~a three-qubit system, and define the quantum channel $\mathcal{E}^\mathrm{df}$ as follows. Let $A\in\mathcal{B}(\mathbb{C}_2^2\otimes\mathbb{C}_3^2)$. Then, for fixed unitary operators $U_0,U_1\in\mathcal{B}(\mathbb{C}_2^2)$ and fixed density operators $\rho_0,\rho_1\in\mathcal{T}(\mathbb{C}_3^2)$, define the quantum channel $\mathcal{E}^\mathrm{df}$ by $\mathcal{E}^\mathrm{df}(\ket{0}\bra{0}\otimes A)=\ket{1}\bra{1}\otimes U_1\tr_3(A)U_1^\dagger\otimes\rho_1$, $\mathcal{E}^\mathrm{df}(\ket{1}\bra{1}\otimes A)=\ket{0}\bra{0}\otimes U_0\tr_3(A)U_0^\dagger\otimes\rho_0$, $\mathcal{E}^\mathrm{df}(\ket{0}\bra{1}\otimes A)=0$, and $\mathcal{E}^\mathrm{df}(\ket{1}\bra{0}\otimes A)=0$. Here, $\mathcal{B}(\mathbb{C}_2^2)$ is a decoherence-free subsystem, on which $\mathcal{E}^\mathrm{df}$ still acts unitarily. In the decomposition of the space of recurrences, we have two non-zero sub-blocks of $\mathbb{C}_2^2\otimes\mathbb{C}_3^2$, which are indexed by $\ket{0}\bra{0}$ and $\ket{1}\bra{1}$ on $\mathcal{B}(\mathbb{C}_1^2)$. The channel $\mathcal{E}^\mathrm{df}$ permutes these two sub-blocks.
	\end{enumerate}
\end{example}

In the next subsection, we will look at a very special class of quantum channels, namely \emph{ergodic quantum channels}.
As we will see, the ergodicity is a spectral property, so the acquired language from this subsection will help to understand this feature better.

\subsection{Ergodic Quantum Channels}\label{sec:ergodic_channels}
The goal of this section is to clarify the spectral structure of the peripheral part of an ergodic quantum channel.
For concreteness, we interlude our analysis with the introduction of some examples, which explicitly demonstrate the features we prove.
Let us start by introducing what ergodic quantum channels are.
\begin{definition}[Ergodic and mixing quantum channels]\label{def:ergodic_mixing}
	If the fixed point of a quantum channel $\mathcal{E}:\mathcal{B}(\mathscr{H})\rightarrow\mathcal{B}(\mathscr{H})$ is unique, i.e.~if $\dim\mathcal{F}(\mathcal{E})=1$, the channel $\mathcal{E}$ is called ``ergodic.'' If, in addition, the unique fixed point is the only element of the space of recurrences of $\mathcal{E}$, i.e.~if $\dim\mathcal{X}(\mathcal{E})=1$, the channel $\mathcal{E}$ is called ``mixing.''
\end{definition}

For a detailed analysis of ergodic and mixing quantum channels, see e.g.~Refs.~\cite{Burgarth2007, burgarth_ergodic_2013}\@.
For further information, we refer to the literature on measure-preserving dynamical systems~\cite{Krengel1985, Eisner2015}, where the weak mixing (strong mixing) property of an endomorphism on a probability space corresponds to ergodicity (mixing) in Definition~\ref{def:ergodic_mixing}\@.
In particular, see Refs.~\cite{Chan2001, Moothathu2009, Burke2016} for relationships of these mathematical concepts to quantum channels.
Notice that every mixing channel is ergodic but the converse is not true, as the following examples show.
\begin{example}[Ergodic and mixing quantum channels]\label{ex:ergodic_mixing}
The following examples show mixing, ergodic, and non-ergodic channels. They shall give some intuition for the above definitions.
	\begin{enumerate}[(i)]
		\item\label{it:example_mixing} Let $\rho_*\in\mathcal{T}(\mathscr{H})$. Then, the quantum channel $\mathcal{P}^{\rho_*}(\rho)=\tr(\rho)\rho_*$ is mixing, with $\mathcal{F}(\mathcal{P}^{\rho_*})=\mathcal{X}(\mathcal{P}^{\rho_*})=\{a\rho_*:a\in\mathbb{C}\}$, which is one-dimensional. Notice that, for $\rho_*=\mathbb{1}/d$, the channel $\mathcal{P}^{\mathbb{1}/d}$ is the completely depolarizing channel, which naturally appeared in our discussion on the system DD in Sec.~\ref{sec:system_DD}\@.
		\item\label{it:example_erg2} Consider the qubit channel $\mathcal{E}^{\updownarrow}(\rho)=\ket{0}\bra{1}\rho\ket{1}\bra{0}+\ket{1}\bra{0}\rho\ket{0}\bra{1}$ from Refs.~\cite{Burgarth2007, burgarth_ergodic_2013}\@. It has a unique fixed-point state $\rho_*^{\updownarrow}=(\ket{0}\bra{0}+\ket{1}\bra{1})/2=\mathbb{1}/2$. Therefore, it is ergodic with $\mathcal{F}(\mathcal{E}^{\updownarrow})=\{a\mathbb{1}:a\in\mathbb{C}\}$. However, it is not mixing, since it has an eigenoperator $\vartheta^{\updownarrow}=\ket{0}\bra{0}-\ket{1}\bra{1}=Z$ corresponding to eigenvalue $\lambda=-1$ in its space of recurrences $\mathcal{X}(\mathcal{E}^{\updownarrow})$. Its peripheral spectrum consists of $\sigma_\varphi(\mathcal{E}^\updownarrow)=\{1,-1\}$, and $\mathcal{X}(\mathcal{E}^{\updownarrow})=\Span\{\mathbb{1},Z\}=\Span\{\ket{0}\bra{0},\ket{1}\bra{1}\}$, which is two-dimensional.
		\item\label{it:example_erg3_1} Consider the qutrit channel $\mathcal{E}^\Lsh(\rho)=\ket{0}\bra{1}\rho\ket{1}\bra{0}+\ket{1}\bra{0}\rho\ket{0}\bra{1}+\ket{0}\bra{2}\rho\ket{2}\bra{0}$. It has a unique fixed-point state $\rho_*^\Lsh=(\ket{0}\bra{0}+\ket{1}\bra{1})/2$. Therefore, it is ergodic with $\mathcal{F}(\mathcal{E}^\Lsh)=\{a\rho_*^\Lsh:a\in\mathbb{C}\}$. However, it is not mixing, since it has an eigenoperator $\vartheta^\Lsh=\ket{0}\bra{0}-\ket{1}\bra{1}$ corresponding to eigenvalue $\lambda=-1$ in its space of recurrences $\mathcal{X}(\mathcal{E}^\Lsh)$. Its peripheral spectrum consists of $\sigma_\varphi(\mathcal{E}^\Lsh)=\{1,-1\}$, and $\mathcal{X}(\mathcal{E}^\Lsh)=\Span\{\rho_*^\Lsh,\vartheta^\Lsh\}=\Span\{\ket{0}\bra{0},\ket{1}\bra{1}\}$, which is two-dimensional.
		\item\label{it:example_erg3_2} Consider the qutrit channel $\mathcal{E}^{\triangle}(\rho)=\ket{0}\bra{1}\rho\ket{1}\bra{0}+\ket{1}\bra{2}\rho\ket{2}\bra{1}+\ket{2}\bra{0}\rho\ket{0}\bra{2}$. It has a unique fixed-point state $\rho_*^{\triangle}=(\ket{0}\bra{0}+\ket{1}\bra{1}+\ket{2}\bra{2})/3=\mathbb{1}/3$. Therefore, it is ergodic with $\mathcal{F}(\mathcal{E}^\triangle)=\{a\rho_*^\triangle:a\in\mathbb{C}\}$. However, it is not mixing, since it has eigenoperators $\vartheta_\pm^\triangle=\ket{0}\bra{0}+\rme^{\mp2\pi\rmi/3}\ket{1}\bra{1}+\rme^{\pm2\pi\rmi/3}\ket{2}\bra{2}$ corresponding to eigenvalues $\lambda=\rme^{\pm2\pi\rmi/3}$ in its space of recurrences $\mathcal{X}(\mathcal{E}^\triangle)$. Its peripheral spectrum consists of $\sigma_\varphi(\mathcal{E}^\triangle)=\{1,\rme^{2\pi\rmi/3},\rme^{-2\pi\rmi/3}\}$, and $\mathcal{X}(\mathcal{E}^\triangle)=\Span\{\rho_*^\triangle,\vartheta_+^\triangle,\vartheta_-^\triangle\}=\Span\{\ket{0}\bra{0},\ket{1}\bra{1},\ket{2}\bra{2}\}$, which is three-dimensional.
		\item\label{it:example_erg3_3} Consider the qutrit channel $\mathcal{E}^\square(\rho)=\sum_{i=1}^4K_i\rho K_i^\dag$, with four Kraus operators
\begin{equation}
K_1
=\begin{pmatrix}
0&0&0\\
0&0&0\\
1&0&0
\end{pmatrix},
\ \ %
K_2
=\begin{pmatrix}
0&0&0\\
0&0&0\\
0&1&0
\end{pmatrix},
\ \ %
K_3
=\sqrt{p}
\begin{pmatrix}
0&0&1\\
0&0&0\\
0&0&0
\end{pmatrix},
\ \ %
K_4
=\sqrt{1-p}
\begin{pmatrix}
0&0&0\\
0&0&1\\
0&0&0
\end{pmatrix},
\end{equation}
where $p\in(0,1)$. 
It swaps two states
\begin{equation}
\rho_0^\square
=\begin{pmatrix}
p&0&0\\
0&1-p&0\\
0&0&0
\end{pmatrix},
\qquad
\rho_1^\square
=\begin{pmatrix}
0&0&0\\
0&0&0\\
0&0&1
\end{pmatrix},
\end{equation}
as $\mathcal{E}^\square(\rho_0^\square)=\rho_1^\square$ and $\mathcal{E}^\square(\rho_1^\square)=\rho_0^\square$, and its space of recurrences reads $\mathcal{X}(\mathcal{E}^\square)=\Span\{\rho_0^\square,\rho_1^\square\}$, with a non-trivial block structure with $d_1=d_2=1$ and $\dim\mathscr{H}_{1,2}=2$, $\dim\mathscr{H}_{2,2}=1$ in the decomposition~\eqref{eq:recurrent_space}.
Its peripheral spectrum consists of $\sigma_\varphi(\mathcal{E}^\square)=\{1,-1\}$, with the peripheral eigenoperators $\rho_*^\square=(\rho_0^\square+\rho_1^\square)/2$ and $\vartheta^\square=\rho_0^\square-\rho_1^\square$ corresponding to the peripheral eigenvalues $\lambda=1$ and $-1$, respectively. 
It has a unique fixed-point state $\rho_*^\square$, and therefore, it is ergodic with $\mathcal{F}(\mathcal{E}^\square)=\{a\rho_*^\square:a\in\mathbb{C}\}$, but it is not mixing.
Its space of recurrences is given by $\mathcal{X}(\mathcal{E}^\square)=\Span\{\rho_*^\square,\vartheta^\square\}=\Span\{\rho_0^\square,\rho_1^\square\}$, which is two-dimensional.
\item\label{it:example_dephasing} An interesting example of quantum channel that is neither ergodic nor mixing is the completely dephasing channel $\mathcal{E}^\mathrm{d}$. For a $d$-dimensional $\mathscr{H}$, it acts on basis operators of $\mathcal{B}(\mathscr{H})$ as $\mathcal{E}^\mathrm{d}(\ket{i}\bra{j})=\delta_{ij}\ket{i}\bra{i}$ ($i,j=0,\ldots,d-1$). The fixed points of $\mathcal{E}^\mathrm{d}$ are given by convex combinations of the states of the form $\ket{i}\bra{i}$ ($i=0,\ldots,d-1$). Therefore, $\mathcal{F}(\mathcal{E}^\mathrm{d})=\Span\{\ket{i}\bra{i}:i=0,\ldots,d-1\}$, which is $d$-dimensional.
	\end{enumerate}
\end{example}
Several criteria for mixing channels are presented in Refs.~\cite{Burgarth2007, burgarth_ergodic_2013}\@.
In addition, it is interesting to note that, for infinitely divisible channels~\cite{Chruscinski2022, Amato2023}, including Markovian CPTP semi-groups generated by Gorini-Kossakowski-Lindblad-Sudarshan (GKLS) generators~\cite{Alicki2007, Nielsen2010, Wolf2012, Chruscinski2017, Chruscinski2022}, the properties of ergodicity and mixing coincide, i.e.~every ergodic CPTP semi-group is also mixing~\cite{Amato2023, Amato2024a}\@.
For semi-group channels, simple criteria to check the ergodicity exist~\cite{Yoshida2023,Zhang2024}\@.

The following proposition clarifies the spectral properties of ergodic quantum channels and provides the main result of this subsection.
A priori, the definition of ergodicity only concerns the fixed-point spectrum of $\mathcal{E}$.
Nevertheless, we show that the restriction of $\mathcal{E}$ having a unique fixed point suffices to characterize the entire spectral decomposition of its peripheral part.
\begin{proposition}[Spectral structure of ergodic quantum channel]\label{prop:ergodic_peripheral}
Let $\mathcal{E}:\mathcal{B}(\mathscr{H})\rightarrow\mathcal{B}(\mathscr{H})$ be an ergodic quantum channel with peripheral part $\mathcal{E}_\varphi=\sum_{|\lambda_\ell|=1}\lambda_\ell\mathcal{P}_\ell$. 
Then, the peripheral spectrum of $\mathcal{E}$ is given by
\begin{equation}
\sigma_\varphi(\mathcal{E})=\{\rme^{2\pi\rmi\ell/K}\}_{\ell=0,\ldots,K-1},
\label{eq:peripheral_spectrum_ergodic}
\end{equation}
with some $K\in\mathbb{N}$, and the peripheral eigenvalues $\lambda_\ell\in\sigma_\varphi(\mathcal{E})$ are not degenerate.
Moreover, the spectral projection $\mathcal{P}_\ell$ belonging to a peripheral eigenvalue $\lambda_\ell=\omega_\ell=\rme^{2\pi\rmi\ell/K}$ ($\ell=0,\ldots,K-1$) is given by
\begin{equation}
\mathcal{P}_\ell=R_\ell\tr(L_\ell^\dag{}\bullet{}),
\label{eq:rank_1_projections}
\end{equation}
with
\begin{gather}
R_\ell=0\oplus\frac{1}{K}\bigoplus_{k=0}^{K-1}\omega_\ell^k\rho_k,
\label{eq:Rk}\\
L_\ell=M\oplus\bigoplus_{k=0}^{K-1}\omega_\ell^k\mathbb{1}_k,
\label{eq:Lk}
\end{gather}
where $\rho_k\,(>0)$ and $\mathbb{1}_k$ are a positive density matrix and the identity, respectively, on the $k$-th block in the space of recurrences $\mathcal{X}(\mathcal{E})$ decomposed as Eq.~\eqref{eq:recurrent_space} with $d_k=1$, and the blocks are relabeled such that the permutation $\pi$ in the recurrence~\eqref{eq:recurrences_map} under the action of $\mathcal{E}$ simply shifts the blocks as $\pi(k)=k+1\mod K$.
Furthermore, $M$ is a matrix that is not necessarily zero.
\end{proposition}
\begin{proof}
Note that the permutation $\pi$ in the recurrence~\eqref{eq:recurrences_map} under the action of a general quantum channel $\mathcal{E}$ consists of separate cycles, each of which is labeled by $c$ and is a cycle of length $K_c$ [thus, $\sum_cK_c=K$ gives the number of blocks in the decomposition~\eqref{eq:recurrent_space}].
Each cycle $c$ permutes sub-blocks of the same dimension $d_c$.
In Ref.~\cite[Theorem~9]{Wolf2010}, it is proven that the peripheral eigenvalues of a general quantum channel $\mathcal{E}$ are given by
\begin{equation}
\{\mu_i^{(c)}\overline{\mu}_j^{(c)}\rme^{2\pi\rmi m/K_c}\}_{c;\,i,j=1,\ldots,d_c;\,m=0,\ldots,K_c-1},
\label{eq:peripheral_spectrum}
\end{equation}
where $\mu_i^{(c)}$ are some phases, i.e.~$|\mu_i^{(c)}|=1$ ($i=1,\ldots,d_c$).
There are $\sum_cK_cd_c^2$ eigenvalues and they exhaust all the peripheral eigenvalues of $\mathcal{E}$ with algebraic multiplicities.
Observe now that $i=j\,(=1,\ldots,d_c)$ with $m=0$ for every cycle $c$ yield generally degenerate eigenvalue $1$.
The assumption that the channel $\mathcal{E}$ is ergodic (i.e.~with non-degenerate eigenvalue $1$) thus requires that $d_c=1$ with a unique cycle.
This proves that the peripheral spectrum of an ergodic quantum channel $\mathcal{E}$ is given by Eq.~\eqref{eq:peripheral_spectrum_ergodic}, and the peripheral eigenvalues are not degenerate.
It is easy to check that the operator $R_k$ defined in Eq.~\eqref{eq:Rk} is the eigenoperator of the ergodic quantum channel $\mathcal{E}$ belonging to the eigenvalue $\omega_\ell$, i.e.
\begin{equation}
\mathcal{E}(R_\ell)=\omega_\ell R_\ell,\quad\ell=0,\ldots,K-1,
\end{equation}
if the blocks in the space of recurrences $\mathcal{X}(\mathcal{E})$ are relabeled such that the permutation $\pi$ in the recurrence~\eqref{eq:recurrences_map} under the action of $\mathcal{E}$ simply shifts the blocks as $\pi(k)=k+1\mod K$.
Moreover, we have the orthogonality
\begin{equation}
\tr(L_\ell^\dag R_{\ell'})=\delta_{\ell\ell'},\quad\ell,\ell'=0,\ldots,K-1.
\end{equation}
The peripheral eigenvalues $\omega_\ell$ are non-degenerate, and the corresponding spectral projections $\mathcal{P}_\ell$ are rank-1.
The spectral projections $\mathcal{P}_\ell$ are thus given by Eq.~\eqref{eq:rank_1_projections}.
Notice that the recurrences $X\in\mathcal{X}(\mathcal{E})$ are of the form of Eq.~\eqref{eq:recurrence}; in particular, they may have a zero block on $\mathscr{H}_0$.
Therefore, $L_\ell$ can act as a (generally) non-zero matrix $M$ on this block without affecting the action of $\mathcal{P}_\ell$ onto $\mathcal{X}(\mathcal{E})$.
We also refer to the forthcoming work~\cite{Amato2024b} for an explicit specification.
\end{proof}
Proposition~\ref{prop:ergodic_peripheral} is a generalization of Ref.~\cite[Theorem~3.1]{Evans1978}, which holds for so-called \emph{irreducible} quantum channels (also see Ref.~\cite{Carbone2019}).
They are ergodic quantum channels whose fixed point is of full rank.
We can directly infer the following corollary from Proposition~\ref{prop:ergodic_peripheral}\@.
\begin{corollary}\label{corollary:ergodic_commute}
Let $\mathcal{E}:\mathcal{B}(\mathscr{H})\rightarrow\mathcal{B}(\mathscr{H})$ be an ergodic quantum channel with peripheral spectral projections $\mathcal{P}_\ell=R_\ell\tr(L_\ell^\dagger{}\bullet{})$ as in Proposition~\ref{prop:ergodic_peripheral}\@. Then,
\begin{enumerate}[(i)]
	\item\label{it:ergodic_commute} $[L_\ell^\dagger, R_\ell]=0$, $\forall\ell$, and
	\item\label{it:ergodic_fixed} the fixed-point state $\rho_*\in\mathcal{T}(\mathscr{H})$ of $\mathcal{E}$ is given by $\rho_*=R_0= L_\ell^\dagger R_\ell=\sqrt{R_\ell^\dag R_\ell}$, $\forall\ell$.
\end{enumerate}
\end{corollary}
\begin{proof}
These follow immediately from Proposition~\ref{prop:ergodic_peripheral}\@. 
See also Ref.~\cite[Lemma~6]{Burgarth2007} for the last equality.
\end{proof}

\begin{example}[Spectral structure of ergodic quantum channel]\label{example:peripheral_projections}
Let us look at the ergodic channels introduced in Example~\ref{ex:ergodic_mixing}\@.
	\begin{enumerate}[(i)]
		\item\label{it:peripheral_mixing} The quantum channel $\mathcal{P}^{\rho_*}$ has only one peripheral eigenoperator, namely the fixed-point state $\rho_*$, onto which $\mathcal{P}^{\rho_*}$ projects. Hence, its only peripheral projection is $\mathcal{P}^{\rho_*}=\rho_*\tr(\mathbb{1}{}\bullet{})$ itself. The channel $\mathcal{P}^{\rho_*}$ is irreducible if and only if $\rho_*$ is of full rank.
		\item\label{it:peripheral_qubit} The qubit channel $\mathcal{E}^\updownarrow$ has two peripheral eigenoperators $\rho_*^\updownarrow=\mathbb{1}/2$ and $\vartheta^\updownarrow=Z$ corresponding to eigenvalues $\lambda_0=1$ and $\lambda_1=-1$, respectively, and its peripheral projections are given by $\mathcal{P}_0^\updownarrow=\frac{1}{2} \mathbb{1}\tr(\mathbb{1}{}\bullet{})$ and $\mathcal{P}_1^\updownarrow=\frac{1}{2}Z\tr(Z{}\bullet{})$. The channel $\mathcal{E}^\updownarrow$ is irreducible as its fixed point $\rho_*^\updownarrow=\mathbb{1}/2$ is of full rank.
		\item\label{it:peripheral_qutrit2} The qutrit channel $\mathcal{E}^\Lsh$ has two peripheral eigenoperators $\rho_*^\Lsh=(\ket{0}\bra{0}+\ket{1}\bra{1})/2$ and $\vartheta^\Lsh=\ket{0}\bra{0}-\ket{1}\bra{1}$ corresponding to eigenvalues $\lambda_0=1$ and $\lambda_1=-1$, respectively, and its peripheral projections are given by $\mathcal{P}_0^\Lsh=R_0^\Lsh\tr(L_0^{\Lsh\dag}{}\bullet{})$ with $R_0^\Lsh=L_0^\Lsh/2=\rho_*^\Lsh$, and $\mathcal{P}_1^\Lsh=\frac{1}{2}\vartheta^\Lsh\tr(\vartheta^\Lsh{}\bullet{})$. The channel $\mathcal{E}^\Lsh$ is ergodic but not irreducible as its fixed point $\rho_*^\Lsh=(\ket{0}\bra{0}+\ket{1}\bra{1})/2$ is not of full rank.
		\item The qutrit channel $\mathcal{E}^{\triangle}$ has three peripheral eigenoperators $\rho_*^{\triangle}=(\ket{0}\bra{0}+\ket{1}\bra{1}+\ket{2}\bra{2})/3=\mathbb{1}/3$ and $\vartheta_\pm^\triangle=\ket{0}\bra{0}+\rme^{\mp2\pi\rmi/3}\ket{1}\bra{1}+\rme^{\pm2\pi\rmi/3}\ket{2}\bra{2}$ corresponding to eigenvalues $\lambda_0=1$ and $\lambda_\pm=\rme^{\pm2\pi\rmi/3}$, respectively. Its peripheral projections are given by $\mathcal{P}_0^\triangle=\frac{1}{3}\mathbb{1}\tr(\mathbb{1}{}\bullet{})$ and $\mathcal{P}_\pm^\triangle=\frac{1}{3}\vartheta_\pm^\triangle\tr(\vartheta_\mp^\triangle{}\bullet{})$. The channel $\mathcal{E}^\triangle$ is irreducible as its fixed point $\rho_*^\triangle=\mathbb{1}/3$ is of full rank.
		\item\label{it:peripheral_qutrit1} The qutrit channel $\mathcal{E}^\square$ has two peripheral eigenoperators $\rho_*^\square=(\rho_0^\square+\rho_1^\square)/2$ and $\vartheta^\square=\rho_0^\square-\rho_1^\square$ corresponding to eigenvalues $\lambda_0=1$ and $\lambda_1=-1$, respectively. 
Its peripheral projections are given by $\mathcal{P}_\ell^\square=R_\ell^\square\tr(L_\ell^{\square\dag}{}\bullet{})$, with
\begin{gather}
R_0^\square
=\frac{1}{2}
\begin{pmatrix}
p&0&0\\
0&1-p&0\\
0&0&1
\end{pmatrix},
\qquad
L_0^\square
=\begin{pmatrix}
1&0&0\\
0&1&0\\
0&0&1
\end{pmatrix},\\
R_1^\square
=\frac{1}{2}
\begin{pmatrix}
p&0&0\\
0&1-p&0\\
0&0&-1
\end{pmatrix},
\qquad
L_1^\square
=\begin{pmatrix}
1&0&0\\
0&1&0\\
0&0&-1
\end{pmatrix}.
\end{gather}
The channel $\mathcal{E}^\square$ is irreducible if and only if $p\in(0,1)$.
\end{enumerate}
\end{example}
The examples discussed in this paper are listed in Table~\ref{tab:examples}\@.
\begin{table}
	\caption{\label{tab:examples}Collection of the examples of quantum channels discussed in this paper. See Examples~\ref{example:DFSS} and~\ref{ex:ergodic_mixing} for the definitions of the selected quantum channels. In this table, we summarize some important spectral properties of the channels. DFSS refers to decoherence-free subsystem.}
	\centering
	\begin{tabular}{|c|c|c|c|c|c|}
		\hline
		\bf{Channel} & \bf{Mixing} & \bf{Irreducible} & \bf{Ergodic} & \bf{No DFSS} & \bf{Refer to}\\
		\hhline{:======:}
		$\mathcal{E}^\Omega$ & \no & \no & \no & \no & Example~\ref{example:DFSS}(i)\\
		\hline
		$\mathcal{E}^\mathrm{df}$ & \no & \no & \no & \no & Example~\ref{example:DFSS}(ii)\\
		\hline
		$\mathcal{E}^\mathrm{d}$ & \no & \no & \no & \yes & Example~\ref{ex:ergodic_mixing}(vi)\\
		\hline
		$\mathcal{E}^\Lsh$ & \no & \no & \yes & \yes & Example~\ref{ex:ergodic_mixing}(iii), \ref{example:peripheral_projections}(iii)\\
		\hline
		$\mathcal{E}^\updownarrow$ & \no & \yes & \yes & \yes & Example~\ref{ex:ergodic_mixing}(ii), \ref{example:peripheral_projections}(ii)\\
		\hline
		$\mathcal{E}^\triangle$ & \no & \yes & \yes & \yes & Example~\ref{ex:ergodic_mixing}(iv), \ref{example:peripheral_projections}(iv)\\
		\hline
		$\mathcal{E}^\square$ & \no & \yes & \yes & \yes & Example~\ref{ex:ergodic_mixing}(v), \ref{example:peripheral_projections}(v)\\
		\hline
		$\mathcal{P}^{\rho_*}$ & \yes & \yes\ iff $\rho_*$ full rank& \yes & \yes & Example~\ref{ex:ergodic_mixing}(i), \ref{example:peripheral_projections}(i)\\
		\hline
	\end{tabular}
\end{table}

Importantly, we emphasize that ergodic quantum channels cannot have any decoherence-free subsystem, which is a direct consequence of Proposition~\ref{prop:ergodic_peripheral}\@.
Vice versa, a channel with a decoherence-free subsystem cannot be ergodic.
On the other hand, there are examples of channels which are neither ergodic nor have a decoherence-free subsystem.
See Example~\ref{ex:ergodic_mixing}\eqref{it:example_dephasing}\@.
Such channels are characterized by the following lemma.
On the space of recurrences $\mathcal{X}(\mathcal{E})$, such a channel decomposes into a direct sum of ergodic channels.
\begin{lemma}[Quantum channel without a decoherence-free subsystem]\label{lemma:no_DFS}
 	Let $\mathcal{E}$ be a quantum channel without a decoherence-free subsystem.
 	Let its peripheral spectrum $\sigma_\varphi(\mathcal{E})$ contain at least one eigenvalue with degeneracy strictly larger than $1$ (thus, $\mathcal{E}$ is \emph{not} ergodic according to Proposition~\ref{prop:ergodic_peripheral}).
 	Then, $\mathcal{E}$ acts as a direct sum of ergodic quantum channels $\mathcal{E}_c$ on $\mathcal{X}(\mathcal{E})$. That is, $\mathcal{E}(X)=\bigoplus_c\mathcal{E}_c(X)$ for any $X\in\mathcal{X}(\mathcal{E})$.
\end{lemma}
\begin{proof}
 	Consider the action of $\mathcal{E}$ on $X\in\mathcal{X}(\mathcal{E})$ as per Eq.~\eqref{eq:recurrences_map}.
	Since by assumption $\mathcal{E}$ does not have a decoherence-free subsystem, it immediately follows that $d_k=1$, $\forall k$.
	Then, the peripheral eigenvalues of $\mathcal{E}$ are given by
	\begin{equation}
	\{\rme^{2\pi\rmi m/K_c}\}_{c;\,m=0,\ldots,K_c-1}.
	\end{equation}
	See Eq.~\eqref{eq:peripheral_spectrum}.
	Notice that each cycle $c$ gives rise to the peripheral spectrum 
	\begin{equation}
	\sigma_\varphi^{(c)}(\mathcal{E})=\{\rme^{2\pi\rmi m/K_c}\}_{m=0,\ldots,K_c-1}
	\end{equation}
	without degeneracy, and acts as an ergodic channel $\mathcal{E}_c$.
	The cycles $c$ are disjoint and $\mathcal{E}_c$ with different $c$ act on blocks belonging to different cycles.
	Therefore, $\mathcal{E}(X)$ is a direct sum of $\mathcal{E}_c(X)$.
\end{proof}
Note that such a quantum channel $\mathcal{E}$ without a decoherence-free subsystem does not necessarily act as a direct sum of ergodic quantum channels on operators $A\not\in\mathcal{X}(\mathcal{E})$.

In summary, in this section we have characterized
\begin{enumerate}[(i)]
	\item the spectral structure of an ergodic quantum channel (Proposition~\ref{prop:ergodic_peripheral}), and
	\item the action of a quantum channel without decoherence-free subsystem on its space of recurrences (Lemma~\ref{lemma:no_DFS}).
\end{enumerate}
These results allow us to study the physical problem of bath DD by kicking with quantum channels (instead of unitaries as discussed in Sec.~\ref{sec:system_DD}).
In the next section, we identify the necessary and sufficient conditions for the bath DD to work.

\section{Bath Dynamical Decoupling}\label{sec:bath_DD}
In the bath DD, we are allowed to employ a CPTP operation to frequently kick the bath, instead of cycles of unitary kicks.
This is legitimate because we do not care about the state of the bath at the end of a decoupling sequence.
What matters is the state of the system, and it does not matter even if the coherence of the bath is destroyed by frequent non-unitary kicks.
We have already seen such a possibility in Sec.~\ref{sec:system_DD}, on the basis of the quantum Zeno dynamics induced by general quantum operations~\cite{burgarth_quantum_2020} (see also Refs.~\cite{Becker2021,Moebus2023}).
Let us recapitulate a result from Ref.~\cite{burgarth_quantum_2020}, that is important for our following discussion.
\begin{fact}[Quantum Zeno dynamics induced by general quantum channels]\label{fact:QZD}
Let $\mathcal{E}:\mathcal{B}(\mathscr{H})\rightarrow\mathcal{B}(\mathscr{H})$ be a quantum channel and let $\mathcal{H}=[H,{}\bullet{}]$ be the adjoint representation of a Hamiltonian $H\in\mathcal{B}(\mathscr{H})$. Then, for all $t\in\mathbb{R}$ and for all $n\in\mathbb{N}$,
	\begin{equation}
		\left(\mathcal{E}\rme^{-\rmi\frac{t}{n}\mathcal{H}}\right)^n=\mathcal{E}_\varphi^n\rme^{-\rmi t\mathcal{H}_\mathrm{Z}^\mathcal{E}} +\mathcal{O}(1/n)\quad\mathrm{as}\quad n\to\infty,
		\label{eq:Zeno_Limit}
	\end{equation}
	where $\mathcal{E}_\varphi$ is the peripheral part of $\mathcal{E}$, and
	\begin{equation}
		\mathcal{H}_\mathrm{Z}^\mathcal{E}=\sum_{|\lambda_\ell|=1}\mathcal{P}_\ell\mathcal{H}\mathcal{P}_\ell,
		\label{eq:Zeno_Hamiltonian}
	\end{equation}
	with $\{\mathcal{P}_\ell\}$ the peripheral projections of $\mathcal{E}$.
\end{fact}
	\begin{proof}
		This is proved in Ref.~\cite[Corollary~1]{burgarth_quantum_2020}\@.
	\end{proof}

Fact~\ref{fact:QZD} has later been generalized to the infinite-dimensional Banach space setting in Refs.~\cite{Moebus2019, Bernad2020, Becker2021}\@. 
We make the following definition for bath DD\@.
\begin{definition}[Bath dynamical decoupling to work]\label{def:bath_DD}
	Let $\mathcal{E}_2$ be a quantum channel acting on $\mathcal{B}(\mathscr{H}_2)$ and let $\mathcal{E}=\mathcal{I}_1\otimes\mathcal{E}_2$ be a quantum channel acting on $\mathcal{B}(\mathscr{H})$ with $\mathcal{I}_1$ the identity on $\mathcal{B}(\mathscr{H}_1)$. Then, we say that ``bath dynamical decoupling works,'' if for any Hamiltonian $\mathcal{H}$, which is generally decomposed as Eq.~\eqref{eq:Hamiltonian_action}, the limit evolution as in Eq.~\eqref{eq:Zeno_Limit} is governed by the Zeno Hamiltonian $\mathcal{H}_\mathrm{Z}^{\mathcal{E}}$ of the form
	\begin{equation}
		\mathcal{H}_\mathrm{Z}^{\mathcal{E}}=\left(\mathcal{H}_1+\sum_i c_i \mathpzc{h}_1^{(i)}\right)\otimes\mathcal{I}_2,
		\label{eq:bath_DD_works}
	\end{equation}
	with some $c_i\in\mathbb{R}$, where $\mathpzc{h}_1^{(i)}=[h_1^{(i)},{}\bullet{}]$, and $\mathcal{I}_2$ is the identity map on $\mathcal{B}(\mathscr{H}_2)$.
\end{definition}
The Zeno Hamiltonian $\mathcal{H}_\mathrm{Z}^\mathcal{E}$ in Eq.~\eqref{eq:bath_DD_works} does not contain any interaction between the subsystems $\mathscr{H}_1$ and $\mathscr{H}_2$.
Therefore, the evolution governed by $\mathcal{H}_\mathrm{Z}^{\mathcal{E}}$ is indeed decoupled.
Nevertheless, we allow the system part of the interaction Hamiltonian, i.e.~$h_1^{(i)}$, to survive the DD, as long as it does not couple the system with the bath.
We will see the condition for suppressing even this part to keep solely the system Hamiltonian $\mathcal{H}_1$.

We now prove that \emph{ergodic quantum channels} are exactly the ones which allow the bath DD to work.
\begin{theorem}[Condition for bath dynamical decoupling]\label{thm:bath_DD}
	Bath dynamical decoupling with the quantum channel $\mathcal{E}=\mathcal{I}_1\otimes\mathcal{E}_2$ works, if and only if $\mathcal{E}_2$ is ergodic.
	The coefficients $c_i$ in the decoupled Hamiltonian~\eqref{eq:bath_DD_works} are given by $c_i=\tr(h_2^{(i)}\rho_*)$, where $\rho_*\in\mathcal{T}(\mathscr{H}_2)$ is the unique fixed-point state of $\mathcal{E}_2$, satisfying $\mathcal{E}_2(\rho_*)=\rho_*$.
\end{theorem}
\begin{proof}
		We first show that the ergodicity of $\mathcal{E}_2$ is a sufficient condition for bath DD to work.
		Afterwards, we show that it is also a necessary condition.

		\emph{Sufficiency:} Let $\mathcal{E}_2$ be an ergodic quantum channel.
		Then, by Proposition~\ref{prop:ergodic_peripheral}, all its peripheral projections $\mathcal{P}_\ell$ are one-dimensional and are given in the form $\mathcal{P}_\ell=R_\ell\tr(L_\ell^\dag{}\bullet{})$.
		For each $\ell$,
		\begin{align}
			&[(\mathcal{I}_1\otimes\mathcal{P}_\ell)\mathcal{H}(\mathcal{I}_1\otimes\mathcal{P}_\ell)](A_1\otimes A_2)
			\nonumber\\
			&\qquad
			=(\mathcal{I}_1\otimes\mathcal{P}_\ell)([H,A_1\otimes R_\ell])\tr(L_\ell^\dagger A_2)\nonumber\\
			&\qquad
			=\tr_2\!\left(
			(\mathbb{1}_1\otimes L_\ell^\dagger)[H,A_1\otimes R_\ell]\right)\otimes R_\ell\tr(L_\ell^\dagger A_2)\nonumber\\
			&\qquad
			=\left(\tr_2[(\mathbb{1}_1\otimes L_\ell^\dagger)H(\mathbb{1}_1\otimes R_\ell)]A_1
			-A_1\tr_2[(\mathbb{1}_1\otimes L_\ell^\dagger)(\mathbb{1}_1\otimes R_\ell)H]\right)\otimes R_\ell\tr(L_\ell^\dagger A_2)\nonumber\\
			&\qquad
			=\Bigl[\tr_2[(\mathbb{1}_1\otimes L_\ell^\dagger R_\ell)H],A_1\Bigr]\otimes R_\ell\tr(L_\ell^\dagger A_2)\nonumber\\
			&\qquad
			=\Bigl[\tr_2[(\mathbb{1}_1\otimes\rho_*)H],A_1\Bigr]\otimes\mathcal{P}_\ell(A_2),
		\end{align}
		where we have used Corollary~\ref{corollary:ergodic_commute}\@.
		Summing over $\ell$ yields
		\begin{equation}
			\mathcal{H}_\mathrm{Z}^\mathcal{E}(A_1\otimes A_2)
			=[H_\mathrm{Z}^\mathcal{E},A_1]\otimes\mathcal{P}_\varphi(A_2),
			\label{eq:Zeno_generator_ergodic}
		\end{equation}
		with $H_\mathrm{Z}^\mathcal{E}=\tr_2[(\mathbb{1}_1\otimes\rho_*)H]$ and $\mathcal{P}_\varphi$ the projection onto the peripheral subspace of $\mathcal{E}_2$.
		Using the operator Schmidt decomposition~\eqref{eq:Hamiltonian_decomposition} of the Hamiltonian $H$, we get
		\begin{equation}
		H_\mathrm{Z}^\mathcal{E}=H_1+\sum_ic_ih_1^{(i)},\label{eq:Zeno_Hamiltonian_suff}
		\end{equation}
		with $c_i=\tr(h_2^{(i)}\rho_*)$.
		Here, we removed the term $\tr(H_2\rho_*)\mathbb{1}$ in Eq.~\eqref{eq:Zeno_Hamiltonian_suff} as it does not affect the commutator $[H_\mathrm{Z}^\mathcal{E},A_1]$ in Eq.~\eqref{eq:Zeno_generator_ergodic}\@.
		Thanks to the presence of $\mathcal{E}_\varphi$ in the limit evolution~\eqref{eq:Zeno_Limit}, we are allowed to replace $\mathcal{P}_\varphi$ by $\mathcal{I}_2$ in the Zeno generator $\mathcal{H}_\mathrm{Z}^\mathcal{E}$ in Eq.~\eqref{eq:Zeno_generator_ergodic}, and get the decoupled Zeno Hamiltonian~\eqref{eq:bath_DD_works}.
		Bath DD thus works.

		\emph{Necessity:} To prove the necessity of the ergodicity of $\mathcal{E}_2$, we assume that $\mathcal{E}_2$ is not ergodic and show that decoupling does not work.
		If $\mathcal{E}_2$ is not ergodic, the eigenvalue $\lambda=1$ of $\mathcal{E}_2$ is degenerate.
		This happens if $\mathcal{E}_2$ admits a decoherence-free subsystem with $d_k\ge2$ in the decomposition~\eqref{eq:recurrent_space}, and/or if the permutation $\pi$ in the recurrence~\eqref{eq:recurrences_map} under the action of $\mathcal{E}_2$ consists of two or more disjoint cycles.
		(In the latter case, each cycle leads to an eigenvalue $1$ due to Equation~\eqref{eq:peripheral_spectrum}\@.)

		Let us first consider the former scenario, assuming the existence of a decoherence-free subsystem in the bath under the action of $\mathcal{E}_2$.
		Suppose that a cycle $c$ in the permutation $\pi$ in the recurrence~\eqref{eq:recurrences_map} under $\mathcal{E}_2$ permutes sub-blocks of dimension $d_c\ge2$, with the length of the cycle $c$ given by $K_c\ge1$.
		We just have to find a counterexample of the Hamiltonian $H$ for which bath DD does not work.
		To this end, we first define
		\begin{equation}
			H_2=0\oplus \bigoplus_{k\in c}(H_{k,1}\otimes\mathbb{1}_{k,2})\oplus \bigoplus_{k\not\in c} (0_{k,1}\otimes 0_{k,2}),
		\end{equation}
		acting on $\mathscr{H}_2$, where $H_{k,1}=H_{k,1}^\dag\in\mathcal{B}(\mathscr{H}_{k,1})$, $\forall k\in c$.
		Here, $k\in c$ ($k\not\in c$) means that block $k$ belongs to the cycle $c$ (cycles disjoint from $c$).
		Furthermore, the $0$ again denotes a $d_0\times d_0$ zero-matrix block acting on $\mathscr{H}_0$ (see Fact~\ref{prop:recurrences}).
		Notice that $H_2$ respects the block structure of $\mathcal{X}(\mathcal{E}_2)$ in Eq.~\eqref{eq:recurrent_space}, and acts only on the blocks relevant to the cycle $c$.
		We will hence focus on the blocks $k\in c$ and omit the other blocks in the following computation.
		In addition, we relabel the blocks $k\in c$ such that $\pi(k)=k+1\mod K_c$.
		Now, by properly choosing $H_{k,1}$, we can make the left (right) action of $H_2$ commutative with the map $\mathcal{E}_2$ on $\mathcal{X}(\mathcal{E}_2)$.
		Indeed, for any $X\in\mathcal{X}(\mathcal{E}_2)$,
		\begin{align}
		H_2\mathcal{E}_2(X)-\mathcal{E}_2(H_2X)
		&=\bigoplus_{k\in c}\left(
		H_{k,1}U_kx_{\pi(k)}U_k^\dag
		-U_kH_{\pi(k),1}x_{\pi(k)}U_k^\dag
		\right)\otimes\rho_k\nonumber\\
		&=\bigoplus_{k\in c}
		(H_{k,1}-U_kH_{\pi(k),1}U_k^\dag)U_kx_{\pi(k)}U_k^\dag
		\otimes\rho_k\nonumber\\
		&=0.
		\end{align}
		This can be accomplished by choosing $H_{k,1}$ so that $[H_{0,1},U_0\cdots U_{K_c-1}]=0$ and $H_{\pi(k),1}=U_k^\dag H_{k,1}U_k$ ($k=0,\ldots,K_c-1$).
		Non-trivial $H_{k,1}=H_{k,1}^\dag\not\propto\mathbb{1}_{k,1}$ fulfilling these conditions exist.
		Such a set $\{H_{k,1}\}$ also yields, for the right action of $H_2$,
		\begin{align}
		\mathcal{E}_2(X)H_2-\mathcal{E}_2(XH_2)
		&=\bigoplus_{k\in c}\left(
		U_kx_{\pi(k)}U_k^\dag H_{k,1}
		-U_kx_{\pi(k)}H_{\pi(k),1}U_k^\dag
		\right)\otimes\rho_k\nonumber\\
		&=\bigoplus_{k\in c}
		U_kx_{\pi(k)}U_k^\dag
		(H_{k,1}-U_kH_{\pi(k),1}U_k^\dag)
		\otimes\rho_k\nonumber\\
		&=0.
		\end{align}
		These commutativities also lead to the commutativity of the adjoint action $\mathcal{H}_2=[H_2,{}\bullet{}]$ with the map $\mathcal{E}_2$ on $\mathcal{X}(\mathcal{E}_2)$,
		\begin{equation}
		[\mathcal{H}_2,\mathcal{E}_2](X)
		=[H_2,\mathcal{E}_2(X)]-\mathcal{E}_2([H_2,X])
		=0.
		\end{equation}
		Moreover, these commutativities are inherited by the peripheral projections $\mathcal{P}_\ell$ of $\mathcal{E}_2$ through $\mathcal{P}_\ell=\oint_{\Gamma_\ell}\frac{\rmd z}{2\pi\rmi}(z\mathcal{I}_2-\mathcal{E}_2)^{-1}$, where $\Gamma_\ell$ is a contour on the complex $z$ plane going anti-clockwise around the $\ell$-th peripheral eigenvalue $\lambda_\ell$.
		That is,
		\begin{equation}
		H_2\mathcal{P}_\ell(X)-\mathcal{P}_\ell(H_2X)=0,\qquad
		\mathcal{P}_\ell(X)H_2-\mathcal{P}_\ell(XH_2)=0,\qquad
		[\mathcal{H}_2,\mathcal{P}_\ell](X)=0,
		\end{equation}
		for any $X\in\mathcal{X}(\mathcal{E}_2)$. 
		Then, for $\mathcal{H}=[H_1\otimes H_2,{}\bullet{}]$ with $H_1\in\mathcal{B}(\mathscr{H}_1)$, we get
		\begin{equation}
		[\mathcal{H},\mathcal{I}_1\otimes\mathcal{P}_\ell](A\otimes X)
		=[H_1,A]\otimes[H_2\mathcal{P}_\ell(X)-\mathcal{P}_\ell(H_2X)]
		+AH_1\otimes[\mathcal{H}_2,\mathcal{P}_\ell](X)
		=0.
		\end{equation}
		Recall the adjoint action~\eqref{eq:Hamiltonian_action} on a product operator.
		Due to this commutativity, the Hamiltonian $\mathcal{H}=[H_1\otimes H_2,{}\bullet{}]$ with the specially designed $H_2$ thus survives the Zeno projection with respect to $\mathcal{E}=\mathcal{I}_1\otimes\mathcal{E}_2$ as
		\begin{align}
		\mathcal{H}_\mathrm{Z}^\mathcal{E}
		&=\sum_{|\lambda_\ell|=1}(\mathcal{I}_1\otimes\mathcal{P}_\ell)\mathcal{H}(\mathcal{I}_1\otimes\mathcal{P}_\ell)
		\nonumber\\
		&=\sum_{|\lambda_\ell|=1}\mathcal{H}(\mathcal{I}_1\otimes\mathcal{P}_\ell)
		\nonumber\\
		&=\mathcal{H}(\mathcal{I}_1\otimes\mathcal{P}_\varphi).
		\label{eq:Zeno_generator_nonergodic}
		\end{align}
		Thanks to the presence of $\mathcal{E}_\varphi$ in the limit evolution~\eqref{eq:Zeno_Limit}, we are allowed to remove $\mathcal{I}_1\otimes\mathcal{P}_\varphi$ from the Zeno generator $\mathcal{H}_\mathrm{Z}^\mathcal{E}$ in Eq.~\eqref{eq:Zeno_generator_nonergodic}, and the Zeno Hamiltonian remains $\mathcal{H}_\mathrm{Z}^\mathcal{E}=\mathcal{H}=[H_1\otimes H_2,{}\bullet{}]$.
		The Zeno projection thus fails to suppress the interaction $H_1\otimes H_2$, and bath DD does not work.

		Let us next consider the other scenario for the non-ergodicity of $\mathcal{E}_2$, where there are two or more disjoint cycles in the recurrence~\eqref{eq:recurrences_map} of $\mathcal{E}_2$.
		In this case, it is much easier to construct a counterexample of the Hamiltonian $H$ for which bath DD does not work.
		We focus on two disjoint cycles $c_1$ and $c_2$, and consider
		\begin{equation}
			H_2
			=0\oplus
			\bigoplus_{k\in c_1}a_1(\mathbb{1}_{k,1}\otimes\mathbb{1}_{k,2})\oplus\bigoplus_{k\in c_2}a_2(\mathbb{1}_{k,1}\otimes\mathbb{1}_{k,2})
			\oplus\bigoplus_{k\not\in c_1\cup c_2}(0_{k,1}\otimes0_{k,2}),
		\end{equation}
		with $a_1,a_2\in\mathbb{R}$ and $a_1\neq a_2$.
		Notice that $H_2$ respects the block structure of $\mathcal{X}(\mathcal{E}_2)$ in Eq.~\eqref{eq:recurrent_space}, and acts only on the blocks relevant to the cycles $c_1$ and $c_2$.
		This $H_2$ trivially commutes with $X\in\mathcal{X}(\mathcal{E}_2)$ and $\mathcal{E}_2$ on $\mathcal{X}(\mathcal{E}_2)$: for any $X\in\mathcal{X}(\mathcal{E}_2)$,
		\begin{equation}
		\mathcal{H}_2(X)=0,\quad
		H_2\mathcal{E}_2(X)-\mathcal{E}_2(H_2X)=0,\quad
		\mathcal{E}_2(X)H_2-\mathcal{E}_2(XH_2)=0,\quad
		[\mathcal{H}_2,\mathcal{E}_2](X)=0,
		\end{equation}
		where $\mathcal{H}_2=[H_2,{}\bullet{}]$.
		Then, the same argument as the former scenario applies, and the coupling Hamiltonian $\mathcal{H}=[H_1\otimes H_2,{}\bullet{}]$ with any $H_1\in\mathcal{B}(\mathscr{H}_1)$ survives the Zeno projection, $\mathcal{H}_\mathrm{Z}^\mathcal{E}=\mathcal{H}=[H_1\otimes H_2,{}\bullet{}]$.
		Even though $\mathcal{H}_2$ itself acts trivially on any $X\in\mathcal{X}(\mathcal{E}_2)$ in the space of recurrences $\mathcal{X}(\mathcal{E}_2)$, i.e.~$\mathcal{H}_2(X)=0$, the coupling Hamiltonian $\mathcal{H}=[H_1\otimes H_2,{}\bullet{}]$ are able to generate a correlation between the system and the bath as $\rme^{-\rmi t\mathcal{H}}(A\otimes X)=A\otimes QXQ+\rme^{-\rmi ta_1\mathcal{H}_1}(A)\otimes P_{c_1}XP_{c_1}+\rme^{-\rmi ta_2\mathcal{H}_1}(A)\otimes P_{c_2}XP_{c_2}$, where $P_{c_i}=0\oplus
			\bigoplus_{k\in c_i}(\mathbb{1}_{k,1}\otimes\mathbb{1}_{k,2})\oplus\bigoplus_{k\not\in c_i}(0_{k,1}\otimes0_{k,2})$ ($i=1,2$) and $Q=\mathbb{1}_2-P_{c_1}-P_{c_2}$.
		This completes the proof that bath DD does not work with a non-ergodic channel $\mathcal{E}_2$.
\end{proof}

Theorem~\ref{thm:bath_DD} gives a complete characterization of bath DD by quantum channels.
It establishes a one-to-one correspondence between bath quantum channels $\mathcal{E}_2:\mathcal{B}(\mathscr{H}_2)\rightarrow\mathcal{B}(\mathscr{H}_2)$ and a decoupled Zeno dynamics.
The criterion on $\mathcal{E}_2$ is easy to check: one only needs to verify that the eigenvalue $\lambda=1$ is not degenerate.
This ensures that $\mathcal{E}_2$ is ergodic.
Computationally, this can be done by exploiting the (row) vectorization isomorphism $\mathrm{vec}:\mathrm{CPTP}(\mathscr{H}_2)\rightarrow\mathcal{B}(\mathscr{H}_2\otimes \mathscr{H}_2)$, which maps a quantum channel $\mathcal{E}_2$ to an operator $\hat{\mathcal{E}}_2$ on an enlarged (doubled) Hilbert space $\mathscr{H}_2\otimes \mathscr{H}_2$~\cite{Havel2003, Wood2015, Watrous2018, Chruscinski2022}\@.
The operator $\hat{\mathcal{E}}_2$ admits a matrix representation on a basis of $\mathscr{H}_2\otimes \mathscr{H}_2$.
By diagonalizing the resulting matrix and by checking the uniqueness of the eigenvalue $\lambda=1$, the ergodicity of $\mathcal{E}_2$ can be verified.
More specifically, we can always write down an explicit matrix representation of $\hat{\mathcal{E}}_2$ if the Kraus operators $E_k$~\cite[Ch.~2.2]{Wolf2012} of $\mathcal{E}_2$ are known in some basis.
Given that $\mathcal{E}_2$ has Kraus representation $\mathcal{E}_2(A)=\sum_k E_kAE_k^\dagger$, the matrix $\hat{\mathcal{E}}_2$ can be written as $\hat{\mathcal{E}}_2=\sum_k E_k\otimes \overline{E}_k$.
See for instance Ref.~\cite[Proposition~2.20]{Watrous2018}\@.
Here, the bar denotes complex conjugation on the chosen basis.

Notice that our bath DD is a generalization of the standard unitary group-based DD introduced in Sec.~\ref{sec:system_DD}\@.
In the unitary group-based DD, one aims to mimic the action of the completely depolarizing channel $\mathcal{P}^{\mathbb{1}/d}$ by averaging over the decoupling set $\mathscr{V}$.
The channel $\mathcal{P}^{\mathbb{1}/d}$ is mixing and thus ergodic.
Nevertheless, since there are ergodic quantum channels that are not mixing (see Example~\ref{ex:ergodic_mixing}), Theorem~\ref{thm:bath_DD} extends the class of possible quantum channels for DD\@.

Interestingly, the condition of the ergodicity of a quantum channel is extremely stable.
In fact, any probabilistic combination of an ergodic quantum channel with another arbitrary (not necessarily ergodic) quantum channel is always ergodic~\cite[Theorem~4]{burgarth_ergodic_2013}\@.
Counterintuitively, a probabilistic mixture of an ergodic quantum channel with the identity channel is even always mixing~\cite[Corollary~6]{burgarth_ergodic_2013}\@.
Therefore, it suffices to only apply an ergodic quantum channel probabilistically with a small probability for bath DD to work.

Finally, let us remark on quantum operations which are \emph{not} trace-preserving.
Such maps do not admit a peripheral spectrum, see e.g.~Ref.~\cite[Remark~1]{burgarth_quantum_2020}\@.
Therefore, their Zeno dynamics according to Fact~\ref{fact:QZD} always leads to a trivial Zeno Hamiltonian $\mathcal{H}^\mathcal{E}_\mathrm{Z}=0$.
This is why we only need to consider true quantum channels in Theorem~\ref{thm:bath_DD}\@.

In the next section, we will investigate the conditions, under which CPTP kicks completely suppress a Hamiltonian.
This study will reveal an interesting qualitative difference between bath DD, as introduced in this section, and unitary system DD, as introduced in Sec.~\ref{sec:system_DD}\@.

\section{Zeno Hamiltonian Suppression}\label{sec:dynamical_freezing}
In Sec.~\ref{sec:bath_DD}, we studied the possibility of eliminating an interaction Hamiltonian $H_\mathrm{I}=\sum_i h_1^{(i)}\otimes h_2^{(i)}$ by kicking the bath with a quantum channel $\mathcal{E}=\mathcal{I}_1\otimes\mathcal{E}_2$.
This ultimately led to Theorem~\ref{thm:bath_DD}, which posed the condition of ergodicity on $\mathcal{E}_2$.
Similarly, for a single system $\mathscr{H}$, instead of the system-bath setting, we would be able to eliminate an arbitrary Hamiltonian $H\in\mathcal{B}(\mathscr{H})$ of the system by frequent applications of a quantum channel $\mathcal{E}$.
This can be regarded as a generalization of the ``quantum Zeno effect,'' which, in its simplest form, \emph{freezes} a system via frequent projective measurements~\cite{Misra1977, Facchi2001a, Facchi2008}\@.
In Ref.~\cite{burgarth_quantum_2020}, it is shown that ``quantum Zeno dynamics'' is induced by \emph{general quantum operations}.
Fact~\ref{fact:QZD} is based on this result.
While transitions among subspaces are suppressed, dynamics within each subspace (called ``Zeno subspace'') is allowed, which is called ``quantum Zeno dynamics''~\cite{Facchi2002, Facchi2008}\@.
Here, in this section, we are going to clarify the condition on the quantum channel $\mathcal{E}$ that further suppresses the quantum Zeno dynamics, by eliminating the Hamiltonian $H$ of the system and achieving the vanishing Zeno Hamiltonian $\mathcal{H}_\mathrm{Z}^\mathcal{E}=0$ via frequent applications of $\mathcal{E}$.

Let us first state what we mathematically mean by ``Zeno Hamiltonian suppression'' in the following discussion.
\begin{definition}[Zeno Hamiltonian suppression]\label{def:dynamical_freezing}
	Let $\mathcal{E}$ be a quantum channel acting on $\mathcal{B}(\mathscr{H})$. Let $H=H^\dag\in\mathcal{B}(\mathscr{H})$ be a Hamiltonian, with adjoint representation $\mathcal{H}=[H,{}\bullet{}]$. Then, we say that ``Zeno Hamiltonian suppression works,'' if the Zeno Hamiltonian $\mathcal{H}_\mathrm{Z}^{\mathcal{E}}$ induced by $\mathcal{E}$ is nullified, i.e.~$\mathcal{H}_\mathrm{Z}^{\mathcal{E}}=0$, for any $\mathcal{H}$.
\end{definition}
Even though the system is not literally frozen in general due to the presence of $\mathcal{E}_\varphi^n$ in the limit evolution~\eqref{eq:Zeno_Limit}, there is no quantum Zeno dynamics within Zeno subspaces if $\mathcal{H}_\mathrm{Z}^{\mathcal{E}}=0$ in the Zeno limit.
If the system starts from a stationary state of $\mathcal{E}$, the system is literally frozen in the initial state.
Such a state actually exists.

We can infer that the condition for Zeno Hamiltonian suppression in Definition~\ref{def:dynamical_freezing} is weaker than the condition for bath DD in Definition~\ref{def:bath_DD}\@.
In fact, it turns out that all quantum channels by which bath DD works also lead to Zeno Hamiltonian suppression.
In the following theorem, we completely characterize the channels by which Zeno Hamiltonian suppression works.
\begin{theorem}\label{thm:dynamical_freezing}
	Zeno Hamiltonian suppression by a quantum channel $\mathcal{E}:\mathcal{B}(\mathscr{H})\rightarrow\mathcal{B}(\mathscr{H})$ works, if and only if $\mathcal{E}$ does not admit any decoherence-free subsystem.
\end{theorem}
\begin{proof}
		Again, we first show the sufficiency and then the necessity.

		\emph{Sufficiency:} Since $\mathcal{E}$ does not admit a decoherence-free subsystem, it is a direct sum of ergodic quantum channels on $\mathcal{X}(\mathcal{E})$ by Lemma~\ref{lemma:no_DFS}\@. 
		Furthermore, all its peripheral projections are given in the form
		\begin{equation}
			\mathcal{P}_\ell=\sum_{j=1}^{\nu_\ell}R_\ell^{(j)}\tr(L_\ell^{(j)\dagger}{}\bullet{}),
		\end{equation}
		where $\nu_\ell$ denotes the degeneracy of the eigenvalue $\lambda_\ell$.
		In addition, $L_\ell^{(i)\dag}$ and $R_\ell^{(j)}$ act on different blocks in the decomposition~\eqref{eq:recurrent_space} of $\mathcal{X}(\mathcal{E})$ if $i\neq j$, i.e.~$L_\ell^{(i)\dag}R_\ell^{(j)}=R_\ell^{(j)}L_\ell^{(i)\dag}=0$ for $i\neq j$.
		Note that $\nu_\ell$ is smaller or equal to the number of cycles in the recurrence~\eqref{eq:recurrences_map} of $\mathcal{E}$.
		Then, the Zeno projection~\eqref{eq:Zeno_Hamiltonian} of the Hamiltonian $\mathcal{H}=[H,{}\bullet{}]$ results in
		\begingroup
		\allowdisplaybreaks
		\begin{align}
			\mathcal{H}_\mathrm{Z}^\mathcal{E}(A)
			&=\sum_{|\lambda_\ell|=1}(\mathcal{P}_\ell\mathcal{H}\mathcal{P}_\ell)(A)\nonumber\\
			&=\sum_{|\lambda_\ell|=1}\left(
			\sum_{i=1}^{\nu_\ell}R_\ell^{(i)}\tr(L_\ell^{(i)\dag}{}\bullet{})
			\right)\mathcal{H}\left(
			\sum_{j=1}^{\nu_\ell}R_\ell^{(j)}\tr(L_\ell^{(j)\dag}A)
			\right)\nonumber\\
			&=\sum_{|\lambda_\ell|=1}
			\sum_{i=1}^{\nu_\ell}
			\sum_{j=1}^{\nu_\ell}
			R_\ell^{(i)}\tr\!\left(
			L_\ell^{(i)\dag}[H,R_\ell^{(j)}]
			\right)\tr(L_\ell^{(j)\dag}A)
			\nonumber\\
			&=\sum_{|\lambda_\ell|=1}
			\sum_{j=1}^{\nu_\ell}
			R_\ell^{(j)}\tr\!\left(
			[R_\ell^{(j)},L_\ell^{(j)\dag}]H
			\right)\tr(L_\ell^{(j)\dag}A)
			\nonumber\\
			&=0,
		\end{align}
		\endgroup
		where we have used Corollary~\ref{corollary:ergodic_commute}\@.

		\emph{Necessity:} We assume that $\mathcal{E}$ admits a decoherence-free subsystem.
		Then, in exactly the same way as in the proof of Theorem~\ref{thm:bath_DD}, we can construct a non-trivial Hamiltonian $H$ which survives the Zeno projection.
		It means that there exists a Hamiltonian $H$ for which Zeno Hamiltonian suppression does not work.
		Therefore, decoherence-free subsystems are not allowed in the kick $\mathcal{E}$ in order for Zeno Hamiltonian suppression to work.
\end{proof}

Theorem~\ref{thm:dynamical_freezing} gives rise to an important qualitative difference between bath DD and the standard group-based unitary DD\@.
In the latter, both goals, the suppression of system-bath interactions and the suppression of a Zeno Hamiltonian, are accomplished by the same scheme: the unitaries have to form a quantum $1$-design, as discussed in Sec.~\ref{sec:system_DD}\@.
The difference comes from the fact that, on one hand, for interaction Hamiltonians of the form $H=H_1\otimes H_2$ the adjoint representation $\mathcal{H}=[H,{}\bullet{}]$ reads
\begin{equation}
	\mathcal{H}(\rho_1\otimes\rho_2)=[H_1,\rho_1]\otimes H_2\rho_2+\rho_1H_1\otimes [H_2,\rho_2],
	\label{eq:interaction_H}
\end{equation}
while, on the other hand, a simple bath Hamiltonian $H=\mathbb{1}_1\otimes H_2$ yields the adjoint representation $\mathcal{H}=[H,{}\bullet{}]$ as
\begin{equation}
	\mathcal{H}(\rho_1\otimes\rho_2)=\rho_1\otimes[H_2,\rho_2].\label{eq:bath_H}
\end{equation}
Therefore, in the case of a simple bath Hamiltonian~\eqref{eq:bath_H}, the DD scheme only has to remove the commutator $[H_2,\rho_2]=H_2\rho_2-\rho_2H_2$.
On the contrary, in the interaction Hamiltonian case~\eqref{eq:interaction_H}, the DD scheme additionally has to remove the individual term $H_2\rho_2$.
This is always ensured in standard unitary group-based DD, where the decoupling scheme is a twirl over the quantum unitary $1$-design $\mathscr{V}$.
Hence, it acts as a projection onto the average of the unitary group, which is a mixing quantum channel.
This sets apart DD from the quantum Zeno dynamics:
while for unitary DD these two concepts are unified~\cite{Facchi2004, Hahn2022}, in the case of CPTP kicks the quantum Zeno effect (Zeno Hamiltonian suppression) is achieved by a more general class of channels than (bath) DD\@.

\section{Examples}\label{sec:examples}
Since the above discussions are rather abstract, we provide some explicit examples in this section.
The channels we study here are the ones that appeared in Examples~\ref{example:DFSS}, \ref{ex:ergodic_mixing}, and~\ref{example:peripheral_projections}, which are collected in Table~\ref{tab:examples}\@.

To numerically investigate the bath DD and the Zeno Hamiltonian suppression, we look at the Choi-Jamio{\l}kowski states~\eqref{eq:Choi-Jamiolkowski} of the respective quantum evolutions.
For the bath DD, given a quantum channel $\mathcal{E}$ acting on $\mathcal{B}(\mathscr{H}_2)$ and an arbitrary Hamiltonian $\mathcal{H}$ acting on $\mathcal{B}(\mathscr{H}_1\otimes\mathscr{H}_2)$, we define the dynamical decoupling evolution
\begin{equation}
	 \mathcal{E}_\mathrm{DD}\equiv\left((\mathcal{I}_1\otimes\mathcal{E})\rme^{-\rmi\frac{t}{n}\mathcal{H}}\right)^n ,
\end{equation}
with $\mathcal{I}_1$ the identity map on $\mathcal{B}(\mathscr{H}_1)$.
For the Zeno Hamiltonian suppression, given a quantum channel $\mathcal{E}$ acting on $\mathcal{B}(\mathscr{H})$ and an arbitrary Hamiltonian $\mathcal{H}$ acting on $\mathcal{B}(\mathscr{H})$, we define the Zeno evolution
\begin{equation}
	 \mathcal{E}_\mathrm{Z}\equiv\left(\mathcal{E}\rme^{-\rmi\frac{t}{n}\mathcal{H}}\right)^n.
\end{equation}

To assess the decoupling fidelity of the bath DD, we compute the purity 
\begin{equation}
\mathtt{P}(\Lambda_1^\mathrm{DD})=\tr(\Lambda_1^{\mathrm{DD}\dagger}\Lambda_1^\mathrm{DD})
\end{equation}
of the reduced Choi-Jamio{\l}kowski state
\begin{equation}
\Lambda_1^\mathrm{DD}\equiv \tr_2(\Lambda(\mathcal{E}_\mathrm{DD}))
\end{equation}
of the decoupling evolution $\mathcal{E}_\mathrm{DD}$.
This purity $\mathtt{P}(\Lambda_1^\mathrm{DD})$ is a measure of the decoupling fidelity since it is directly related to the distance of the system evolution on $\mathscr{H}_1$ to a unitary evolution.
See Ref.~\cite[Proposition~12]{Hahn2022}\@.

To assess how well the Zeno Hamiltonian suppression works, we compute the distance
\begin{equation}
\Vert \Lambda(\mathcal{E}_\mathrm{Z})-\Lambda(\mathcal{E}_\varphi^n)\Vert_1
\end{equation}
between the Choi-Jamio{\l}kowski states of $\mathcal{E}_\mathrm{Z}$ and of the target evolution $\mathcal{E}_\varphi^n$ without the Zeno Hamiltonian $\mathcal{H}_\mathrm{Z}^\mathcal{E}$.
Here, $\Vert A\Vert_1=\tr\Bigl(\sqrt{A^\dagger A}\Bigr)$ is the trace norm.
We remark that the quantity $\Vert \Lambda(\mathcal{E}_\mathrm{Z})-\Lambda(\mathcal{E}_\varphi^n)\Vert_1$ is equivalent to the diamond norm distance between the channels $\mathcal{E}_\mathrm{Z}$ and $\mathcal{E}_\varphi^n$~\cite[Lemma~26]{Hahn2022}\@.

To see the performance of the bath DD (the Zeno Hamiltonian suppression), we compute $\mathcal{E}_\mathrm{DD}$ ($\mathcal{E}_\mathrm{Z}$) for 100 generators $\mathcal{H}=[H,{}\bullet{}]$ generated by randomly sampling 100 Hamiltonians $H$ normalized as $\Vert H\Vert_\infty=1$, with the operator norm $\Vert A\Vert_\infty=\sup_{\Vert \psi\Vert=1} \Vert A\psi\Vert$ giving the largest singular value of $A$.
For the bath DD, we look at the worst-case (minimum) purity $\min\mathtt{P}(\Lambda_1^\mathrm{DD})$ among those for the randomly generated Hamiltonians $H\in\mathcal{B}(\mathscr{H}_1\otimes\mathscr{H}_2)$.
The convergence of this purity to $1$ shows that the bath DD works.
For the Zeno Hamiltonian suppression, on the other hand, we look at the worst-case (maximum) Choi-Jamio{\l}kowski state distance $\max\Vert \Lambda(\mathcal{E}_\mathrm{Z})-\Lambda(\mathcal{E}_\varphi^n)\Vert_1$ among those for the randomly generated Hamiltonians $H\in\mathcal{B}(\mathscr{H})$.
If it goes to zero as $n\rightarrow\infty$, we can conclude that the Zeno Hamiltonian suppression works irrespective of the Hamiltonian $H$.
To see the failure of the bath DD (the Zeno Hamiltonian suppression), we construct a Hamiltonian $H$ for which the purity $\mathtt{P}(\Lambda_1^\mathrm{DD})$ (the distance $\Vert \Lambda(\mathcal{E}_\mathrm{Z})-\Lambda(\mathcal{E}_\varphi^n)\Vert_1)$ remains far from $1$ ($0$) even for large $n$.
Furthermore, we show the average purity $\mathbb{E}\big[\mathtt{P}(\Lambda_1^\mathrm{DD})\big]$ (the average distance $\mathbb{E}\big[\Vert \Lambda(\mathcal{E}_\mathrm{Z})-\Lambda(\mathcal{E}_\varphi^n)\Vert_1)\big]$) over those for the randomly generated Hamiltonians $H\in\mathcal{B}(\mathscr{H}_1\otimes\mathscr{H}_2)$ ($H\in\mathcal{B}(\mathscr{H})$) and show that it remains far from $1$ ($0$) even for large $n$.

\subsection{Ergodic Qubit Channel $\mathcal{E}^\updownarrow$}\label{sec:example_ergodic_qubit}
The qubit channel $\mathcal{E}^\updownarrow$ defined in Example~\ref{ex:ergodic_mixing}\eqref{it:example_erg2} acts on a density operator $\rho\in\mathcal{T}(\mathbb{C}^2)$ as $\mathcal{E}^{\updownarrow}(\rho)=\ket{0}\bra{1}\rho\ket{1}\bra{0}+\ket{1}\bra{0}\rho\ket{0}\bra{1}$, and it is ergodic.
As discussed in Example~\ref{example:peripheral_projections}\eqref{it:peripheral_qubit}, its peripheral projections are given by $\mathcal{P}_0^\updownarrow=\frac{1}{2}\mathbb{1}\tr(\mathbb{1}{}\bullet{})$ and $\mathcal{P}_1^\updownarrow=\frac{1}{2}Z\tr(Z{}\bullet{})$, corresponding to eigenvalues $\lambda_0=1$ and $\lambda_1=-1$, respectively, with $Z$ the third Pauli matrix.

To see the bath DD with this qubit channel $\mathcal{E}^\updownarrow$, we consider $\mathscr{H}_1=\mathbb{C}_1^2$ so that the total system-bath Hilbert space is $\mathscr{H}=\mathbb{C}_1^2\otimes\mathbb{C}_2^2$, and apply the bath DD scheme $\mathcal{E}_\mathrm{DD}^\updownarrow=\left((\mathcal{I}_1\otimes\mathcal{E}^\updownarrow)\rme^{-\rmi\frac{t}{n}\mathcal{H}}\right)^n$ for a dephasing interaction of the form $H=\sigma_{n_1}\otimes\sigma_{n_2}\in\mathcal{B}(\mathbb{C}_1^2\otimes\mathbb{C}_2^2)$, with adjoint representation $\mathcal{H}=[H,{}\bullet{}]$, where $\sigma_{n_i}=\bm{n}_i\cdot\bm{\sigma}$ with $\bm{n}_i\in\mathbb{R}^3$ ($i=1,2$), and the elements of $\bm{\sigma}=(X,Y,Z)$ are the Pauli matrices.
If we manage to remove this type of interaction, the bath DD works for general system-bath interactions.
Recall the operator Schmidt decomposition in Eq.~\eqref{eq:Hamiltonian_decomposition}.
In this case, the Zeno Hamiltonian $\mathcal{H}^{\mathcal{E}^\updownarrow}_\mathrm{Z}$, acting on a product input $A=A_1\otimes A_2\in\mathcal{B}(\mathbb{C}_1^2\otimes\mathbb{C}_2^2)$, reads
\begin{align}
	\mathcal{H}^{\mathcal{E}^\updownarrow}_\mathrm{Z}(A)
	={}&[(\mathcal{I}_1\otimes\mathcal{P}_0^\updownarrow)\mathcal{H}(\mathcal{I}_1\otimes\mathcal{P}_0^\updownarrow)](A)
	+[(\mathcal{I}_1\otimes\mathcal{P}_1^\updownarrow)\mathcal{H}(\mathcal{I}_1\otimes\mathcal{P}_1^\updownarrow)](A)\nonumber\\
	={}&\frac{1}{2}[(\mathcal{I}_1\otimes\mathcal{P}_0^\updownarrow)\mathcal{H}](A_1\otimes\mathbb{1})\tr(A_2)
	+\frac{1}{2}[(\mathcal{I}_1\otimes\mathcal{P}_1^\updownarrow)\mathcal{H}](A_1\otimes Z)\tr(ZA_2)\nonumber\\
	={}&\frac{1}{2}(\mathcal{I}_1\otimes\mathcal{P}_0^\updownarrow)\Bigl(
	[\sigma_{n_1},A_1]\otimes\sigma_{n_2}
	\Bigr)\tr(A_2)
	\nonumber\\
	&{}
	+\frac{1}{2}(\mathcal{I}_1\otimes\mathcal{P}_1^\updownarrow)\Bigl(
	[\sigma_{n_1},A_1]\otimes\sigma_{n_2}Z+A_1\sigma_{n_1}\otimes[\sigma_{n_2},Z]
	\Bigr)\tr(ZA_2)\nonumber\\
	={}&\frac{1}{4}[\sigma_{n_1},A_1]\otimes\mathbb{1}\tr(\sigma_{n_2})\tr(A_2)
	\nonumber\\
	&{}
	+\frac{1}{4}\,\Bigl(
	[\sigma_{n_1},A_1]\otimes Z\tr(\sigma_{n_2})+A_1\sigma_{n_1}\otimes Z\tr(Z[\sigma_{n_2},Z])
	\Bigr)\tr(ZA_2)\nonumber\\
	={}&0,\vphantom{\frac{1}{4}}
\end{align}
since $\tr(\sigma_{n_2})=\tr(Z[\sigma_{n_2},Z])=0$.
Therefore, the interaction Hamiltonian $H$ is completely cancelled by the bath DD with $\mathcal{E}^\updownarrow$.
The bath DD works.

To see the Zeno Hamiltonian suppression with the qubit channel $\mathcal{E}^\updownarrow$, we consider the Zeno evolution $\mathcal{E}_\mathrm{Z}^\updownarrow\equiv\left(\mathcal{E}\rme^{-\rmi\frac{t}{n}\mathcal{H}}\right)^n$ for the Hamiltonian $\mathcal{H}=[H,{}\bullet{}]$ with $H=\sigma_n=\bm{n}\cdot\bm{\sigma}\in\mathcal{B}(\mathbb{C}^2)$, $\bm{n}\in\mathbb{R}^3$, and compute the Zeno Hamiltonian $\mathcal{H}_\mathrm{Z}^{\mathcal{E}^\updownarrow}$.
For any $A\in\mathcal{B}(\mathbb{C}^2)$, we have
\begin{align}
\mathcal{H}_\mathrm{Z}^{\mathcal{E}^\updownarrow}(A)
={}&(\mathcal{P}_0^\updownarrow\mathcal{H}\mathcal{P}_0^\updownarrow)(A)
+(\mathcal{P}_1^\updownarrow\mathcal{H}\mathcal{P}_1^\updownarrow)(A)\nonumber\\
={}&\frac{1}{2}(\mathcal{P}_0^\updownarrow\mathcal{H})(\mathbb{1})\tr(A)
+\frac{1}{2}(\mathcal{P}_1^\updownarrow\mathcal{H})(Z)\tr(ZA)\nonumber\\
={}&\frac{1}{2}\mathcal{P}_1^\updownarrow([\sigma_n,Z])\tr(ZA)\nonumber\\
={}&\frac{1}{4}Z\tr(Z[\sigma_n,Z])\tr(ZA)\nonumber\\
={}&0.\vphantom{\frac{1}{4}}
\end{align}
The Zeno Hamiltonian suppression works.

We present numerical results on the bath DD and the Zeno Hamiltonian suppression with $\mathcal{E}^\updownarrow$ in Fig.~\ref{fig:Ergodic_Channel}\@.
Figure~\ref{fig:Ergodic_Channel}(a) shows the worst-case purity $\min\mathtt{P}(\Lambda_1^\mathrm{DD})$ of the reduced Choi-Jamio{\l}kowski state $\Lambda_1^\mathrm{DD}=\tr_2(\Lambda(\mathcal{E}_\mathrm{DD}^\updownarrow))$ for 100 randomly sampled Hamiltonians $H\in\mathcal{B}(\mathbb{C}_1^2\otimes\mathbb{C}_2^2)$ with $\|H\|_\infty=1$.
The purity approaches $1$ as the number of decoupling steps $n$ increses, which shows that the bath DD works for any Hamiltonian $H$.
Since the condition for the bath DD to work is stronger than that for the Zeno Hamiltonian suppression, the latter also works with $\mathcal{E}^\updownarrow$.
This is shown in Fig.~\ref{fig:Ergodic_Channel}(b), where we compute the distance $\Vert \Lambda(\mathcal{E}_\mathrm{Z}^\updownarrow)-\Lambda((\mathcal{E}^\updownarrow)^n)\Vert_1$ between the Zeno evolution $\mathcal{E}_\mathrm{Z}^\updownarrow$ and the target evolution $(\mathcal{E}_\varphi^\updownarrow)^n=(\mathcal{E}^\updownarrow)^n$ in terms of their respective Choi-Jamio{\l}kowski states $\Lambda(\mathcal{E}_\mathrm{Z}^\updownarrow)$ and $\Lambda((\mathcal{E}^\updownarrow)^n)$.
Again, we look at the worst-case distance $\max\Vert \Lambda(\mathcal{E}_\mathrm{Z}^\updownarrow)-\Lambda((\mathcal{E}^\updownarrow)^n)\Vert_1$ for 100 randomly sampled Hamiltonians $H\in\mathcal{B}(\mathbb{C}^2)$ with $\|H\|_\infty=1$.
The Zeno Hamiltonian suppression works with a convergence rate of $\mathcal{O}(1/n)$.
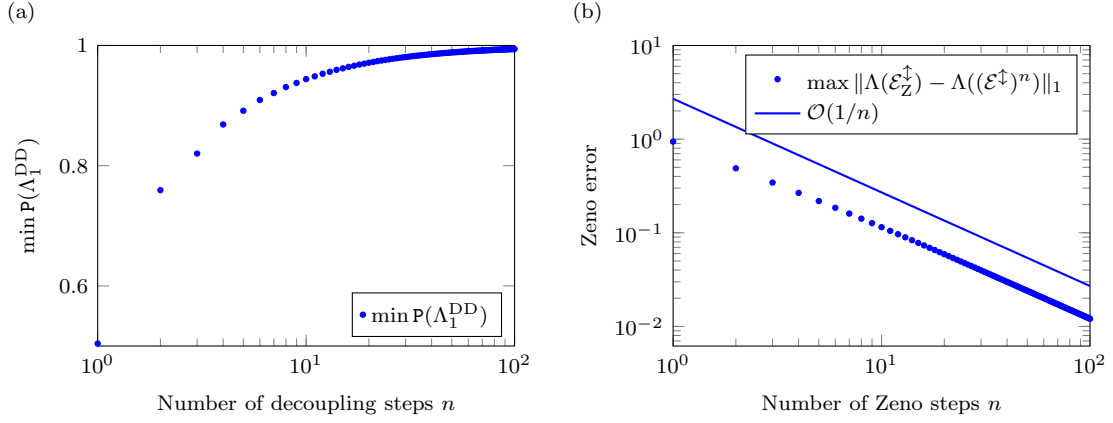
\begin{figure}
\centering
\begin{tabular}{rr}
\multicolumn{1}{l}{\footnotesize(a)}
&
\multicolumn{1}{l}{\footnotesize(b)}
\\
	\begin{tikzpicture}[mark size={1}, scale=1]
	\begin{axis}[
	xmin=1,
	xmax=100,
	ymin=0.5,
	ymax=1,
	ylabel near ticks,
	xlabel={\footnotesize Number of decoupling steps $n$},
	ylabel={\footnotesize $\min\mathtt{P}(\Lambda_1^\mathrm{DD})$},
	x post scale=0.8,
	y post scale=0.7,
	xmode=log,
	tick label style={font=\footnotesize},
	legend style={font=\footnotesize},
	legend pos={south east},
	legend cell align=left,
	]
	\addplot[blue, only marks] table [x=n, y=P, col sep=comma] {Ergodic_Purity_Bi.csv};
	\addlegendentry{$\min\mathtt{P}(\Lambda_1^\mathrm{DD})$};
	\end{axis}
	\end{tikzpicture}
&
  	\begin{tikzpicture}[mark size={1}, scale=1]
	\begin{axis}[
	xmin=1,
	xmax=100,
	ymin=0,
	ymax=10,
	ylabel near ticks,
	xlabel={\footnotesize Number of Zeno steps $n$},
	ylabel={\footnotesize Zeno error},
	xmode=log,
	ymode=log,
	x post scale=0.8,
	y post scale=0.7,
	tick label style={font=\footnotesize},
	legend style={font=\footnotesize},
	legend pos={north east},
	legend cell align=left,
	]
	\addplot[blue,only marks] table [x=n, y=error, col sep=comma] {Ergodic_Error_Mono.csv};
	\addlegendentry{$\max \Vert \Lambda(\mathcal{E}_\mathrm{Z}^\updownarrow)-\Lambda((\mathcal{E}^\updownarrow)^n)\Vert_1$};
	\addplot[color=blue, domain=1:100,thick] {2.7/x};
	\addlegendentry{$\mathcal{O}(1/n)$};
	\end{axis}
	\end{tikzpicture}
\end{tabular}
  	\caption{\label{fig:Ergodic_Channel}Bath dynamical decoupling and Zeno Hamiltonian suppression with the ergodic qubit channel $\mathcal{E}^\updownarrow$ introduced in Example~\ref{ex:ergodic_mixing}(ii). (a) The worst-case purity of the reduced Choi-Jamio{\l}kowski state, $\min\mathtt{P}(\Lambda_1^\mathrm{DD})$, for 100 generators $\mathcal{H}=[H,{}\bullet{}]$ constructed from randomly sampled Hamiltonians $H\in\mathcal{B}(\mathbb{C}_1^2\otimes\mathbb{C}_2^2)$ with $\|H\|_\infty=1$ for the bath dynamical decoupling. (b) The worst-case distance for the Zeno Hamiltonian suppression, $\max\Vert \Lambda(\mathcal{E}_\mathrm{Z}^\updownarrow)-\Lambda((\mathcal{E}^\updownarrow)^n)\Vert_1$, for 100 generators $\mathcal{H}=[H,{}\bullet{}]$ constructed from randomly sampled Hamiltonians $H\in\mathcal{B}(\mathbb{C}^2)$ with $\|H\|_\infty=1$. In both panels, $t=1$.}
\end{figure}

\subsection{Dephasing Channel $\mathcal{E}^\mathrm{d}$}\label{sec:example_dephasing}
The dephasing channel $\mathcal{E}^\mathrm{d}$ is introduced in Example~\ref{ex:ergodic_mixing}\eqref{it:example_dephasing}\@.
It is defined by its action on basis states $\{\ket{i}\}$ of a $d$-dimensional system as $\mathcal{E}^\mathrm{d}(\ket{i}\bra{j})=\delta_{ij}\ket{i}\bra{i}$ ($i,j=0,\ldots,d-1$). 
This channel $\mathcal{E}^\mathrm{d}$ is not ergodic but does not have a decoherence-free subsystem.
Thus, by Lemma~\ref{lemma:no_DFS}, it can be written as a direct sum of ergodic quantum channels on its space of recurrences $\mathcal{X}(\mathcal{E}^\mathrm{d})$.
In fact, $\mathcal{E}^\mathrm{d}=\sum_{i=0}^{d-1}\mathcal{P}_{\ket{i}\bra{i}}$, with $\mathcal{P}_{\ket{i}\bra{i}}(A)=\ket{i}\bra{i}A\ket{i}\bra{i}$ being the peripheral projections of $\mathcal{E}^\mathrm{d}$.
Notice that $\mathcal{E}^\mathrm{d}$ has $d$ spectral cycles, each of length $1$: $\mathcal{E}^\mathrm{d}(\ket{i}\bra{i})=\ket{i}\bra{i}$ ($i=0,\ldots,d-1$).
The space of recurrences $\mathcal{X}(\mathcal{E}^\mathrm{d})$ consists only of fixed points, $\mathcal{X}(\mathcal{E}^\mathrm{d})=\mathcal{F}(\mathcal{E}^\mathrm{d})=\{\ket{i}\bra{i}:i=0,\ldots,d-1\}$.

Let us see how the bath DD fails to work with this channel $\mathcal{E}^\mathrm{d}$.
We consider $\mathscr{H}_1=\mathbb{C}_1^2$ for the system and $\mathscr{H}_2=\mathbb{C}_2^2$ for the bath, so that the total Hilbert space is given by $\mathscr{H}=\mathbb{C}_1^2\otimes\mathbb{C}_2^2$.
We choose the Hamiltonian $H=Z\otimes Z\in\mathcal{B}(\mathbb{C}_1^2\otimes\mathbb{C}_2^2)$, with adjoint representation $\mathcal{H}=[H,{}\bullet{}]$. 
Notice that $Z=\ket{0}\bra{0}-\ket{1}\bra{1}$ acts nontrivially on both of the two cycles.
It is exactly of the structure of the Hamiltonian chosen in the proof of necessity in Theorem~\ref{thm:bath_DD}\@.
We now show that this interaction Hamiltonian $H=Z\otimes Z$ cannot be cancelled by the decoupling kicks with $\mathcal{I}_1\otimes\mathcal{E}^\mathrm{d}$.
Indeed, in this case, the DD evolution reads
\begin{align}
\mathcal{E}_\mathrm{DD}^\mathrm{d}
&=\left((\mathcal{I}_1\otimes\mathcal{E}^\mathrm{d})\rme^{-\rmi\frac{t}{n}\mathcal{H}}\right)^n
\nonumber\\
&=\left(
\rme^{-\rmi\frac{t}{n}\mathcal{H}_1}\otimes\mathcal{P}_{\ket{0}\bra{0}}
+
\rme^{\rmi\frac{t}{n}\mathcal{H}_1}\otimes\mathcal{P}_{\ket{1}\bra{1}}
\right)^n
\nonumber\\
&=\rme^{-\rmi t\mathcal{H}_1}\otimes\mathcal{P}_{\ket{0}\bra{0}}
+
\rme^{\rmi t\mathcal{H}_1}\otimes\mathcal{P}_{\ket{1}\bra{1}},
\end{align}
where $\mathcal{H}_1=[Z,{}\bullet{}]$, which acts on $\mathcal{B}(\mathbb{C}_1^2)$.
The system-bath evolution is not decoupled.
It can induce classical correlations between the system and the bath.
To see this explicitly, consider an arbitrary input density operator $\rho_1\in\mathcal{T}(\mathbb{C}_1^2)$ of the system satisfying $[Z,\rho_1]\neq 0$ and an input density operator $\rho_2=p\ket{0}\bra{0}+(1-p)\ket{1}\bra{1}\in\mathcal{T}(\mathbb{C}_2^2)$ of the bath with $p\in(0,1)$.
Then, the DD evolution $\mathcal{E}_\mathrm{DD}^{\mathrm{d}}$ maps the input product state $\rho=\rho_1\otimes\rho_2\in\mathcal{T}(\mathbb{C}_1^2\otimes\mathbb{C}_2^2)$ to
\begin{equation}
\mathcal{E}_\mathrm{DD}^{\mathrm{d}}(\rho)
=p\rme^{-\rmi tZ}\rho_1\rme^{\rmi tZ}\otimes\ket{0}\bra{0}
+
(1-p)\rme^{\rmi tZ}\rho_1\rme^{-\rmi tZ}\otimes\ket{1}\bra{1},
\end{equation}
developing a classical correlation between the system and the bath.
The bath DD does not work.
In Fig.~\ref{fig:Dephasing}(a), the purity $\mathtt{P}(\Lambda_1^\mathrm{DD})$ of the reduced Choi-Jamio{\l}kowski state $\Lambda_1^\mathrm{DD}=\tr_2(\Lambda(\mathcal{E}_\mathrm{DD}^\mathrm{d}))$ for the Hamiltonian $H=Z\otimes Z\in\mathcal{B}(\mathbb{C}_1^2\otimes\mathbb{C}_2^2)$ and the average purity $\mathbb{E}[\mathtt{P}(\Lambda_1^\mathrm{DD})]$ over 100 randomly sampled Hamiltonians $H\in\mathcal{B}(\mathbb{C}_1^2\otimes\mathbb{C}_2^2)$ with $\|H\|_\infty=1$ are shown as functions of the number of decoupling steps $n$.
The purity stays far from $1$, namely, the DD evolution $\mathcal{E}_\mathrm{Z}^\mathrm{d}$ is far from unitary, no matter how large we choose the number of decoupling steps $n$.
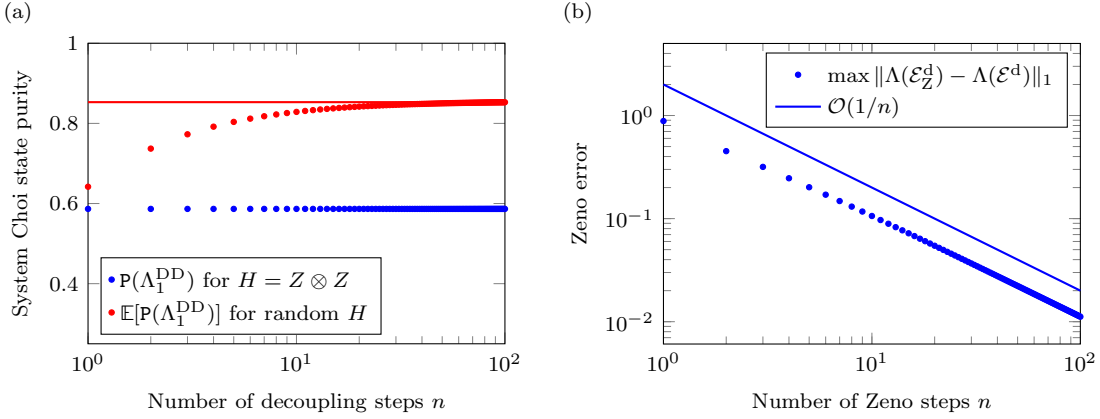
\begin{figure}
\begin{tabular}{rr}
\multicolumn{1}{l}{\footnotesize(a)}
&
\multicolumn{1}{l}{\footnotesize(b)}
\\
	\begin{tikzpicture}[mark size={1}, scale=1]
	\begin{axis}[
	xmin=1,
	xmax=100,
	ymin=0.25,
	ymax=1,
	ylabel near ticks,
	xlabel={\footnotesize Number of decoupling steps $n$},
	ylabel={\footnotesize System Choi state purity},
	x post scale=0.8,
	y post scale=0.7,
	xmode=log,
	tick label style={font=\footnotesize},
	legend style={font=\footnotesize},
	legend pos={south west},
	legend cell align=left,
	]
	\addplot[blue, only marks] table [x=n, y=P, col sep=comma] {Dephasing_Purity_Bi.csv};
	\addplot[red, only marks] table [x=n, y=P, col sep=comma] {Dephasing_Purity_Bi_Rand.csv};
	\addplot[color=red, domain=1:100,thick] {0.853};
	\addlegendentry{$\mathtt{P}(\Lambda_1^\mathrm{DD})$ for $H=Z\otimes Z$};
	\addlegendentry{$\mathbb{E}[\mathtt{P}(\Lambda_1^\mathrm{DD})]$ for random $H$};
	\end{axis}
	\end{tikzpicture}
&
  	\begin{tikzpicture}[mark size={1}, scale=1]
	\begin{axis}[
	xmin=1,
	xmax=100,
	ymin=0,
	ymax=5,
	ylabel near ticks,
	xlabel={\footnotesize Number of Zeno steps $n$},
	ylabel={\footnotesize Zeno error},
	xmode=log,
	ymode=log,
	x post scale=0.8,
	y post scale=0.7,
	tick label style={font=\footnotesize},
	legend style={font=\footnotesize},
	legend pos={north east},
	legend cell align=left,
	]
	\addplot[blue,only marks] table [x=n, y=error, col sep=comma] {Dephasing_Error_Mono.csv};
	\addlegendentry{$\max \Vert \Lambda(\mathcal{E}_\mathrm{Z}^\mathrm{d})-\Lambda(\mathcal{E}^\mathrm{d})\Vert_1$};
	\addplot[color=blue, domain=1:100,thick] {2/x};
	\addlegendentry{$\mathcal{O}(1/n)$};
	\end{axis}
	\end{tikzpicture}
\end{tabular}
\caption{\label{fig:Dephasing}Bath dynamical decoupling and Zeno Hamiltonian suppression with the qubit dephasing channel $\mathcal{E}^\mathrm{d}$ introduced in Example~\ref{ex:ergodic_mixing}(vi)\@. (a) The purity $\mathtt{P}(\Lambda_1^\mathrm{DD})$ of the reduced Choi-Jamio{\l}kowski state for the Hamiltonian $H=Z\otimes Z\in\mathcal{B}(\mathbb{C}_1^2\otimes\mathbb{C}_2^2)$ (blue) and the average purity $\mathbb{E}[\mathtt{P}(\Lambda_1^\mathrm{DD})]$ over 100 randomly sampled Hamiltonians $H\in\mathcal{B}(\mathbb{C}_1^2\otimes\mathbb{C}_2^2)$ normalized as $\|H\|_\infty=1$ (red), for the bath dynamical decoupling. The total evolution time is fixed at $t=1$. The former (blue) is constant at $\mathtt{P}(\Lambda_1^\mathrm{DD})\approx0.59$, while the latter (red) saturates to $\mathbb{E}[\mathtt{P}(\Lambda_1^\mathrm{DD})]\approx0.85$. The bath dynamical decoupling does not work. (b) The worst-case distance for the Zeno Hamiltonian suppression, $\max\Vert \Lambda(\mathcal{E}_\mathrm{Z}^\mathrm{d})-\Lambda(\mathcal{E}^\mathrm{d})\Vert_1$, for 100 randomly sampled Hamiltonians $H\in\mathcal{B}(\mathbb{C}^2)$ normalized as $\|H\|_\infty=1$. The total evolution time is fixed at $t=1$.}
\end{figure}

Even though the bath DD with $\mathcal{E}^\mathrm{d}$ does not work, the Zeno Hamiltonian suppression does.
This is because $\mathcal{E}^\mathrm{d}$ does not have any decoherence-free subsystem.
We can quickly convince ourselves that indeed any Hamiltonian $H\in\mathcal{B}(\mathscr{H})$ with adjoint representation $\mathcal{H}=[H,{}\bullet{}]$ gets switched off by frequently kicking the system with $\mathcal{E}^\mathrm{d}:\mathcal{B}(\mathscr{H})\rightarrow\mathcal{B}(\mathscr{H})$.
For simplicity, let us consider $\mathscr{H}=\mathbb{C}^2$, so that $\mathcal{E}^\mathrm{d}(A)=\ket{0}\bra{0}A\ket{0}\bra{0}+\ket{1}\bra{1}A\ket{1}\bra{1}$, and consider the Zeno evolution $\mathcal{E}_\mathrm{Z}^\mathrm{d}=\Bigl(\mathcal{E}^\mathrm{d}\rme^{-\rmi\frac{t}{n}\mathcal{H}}\Bigr)^n$. 
Since $\mathcal{E}^\mathrm{d}$ itself is the peripheral projection of $\mathcal{E}^\mathrm{d}$, we get the Zeno Hamiltonian
\begin{align}
	\mathcal{H}_\mathrm{Z}^{\mathcal{E}^\mathrm{d}}(\rho)
	={}&(\mathcal{E}^\mathrm{d}\mathcal{H}\mathcal{E}^\mathrm{d})(\rho)\nonumber\\
	={}&(\mathcal{E}^\mathrm{d}\mathcal{H})(\ket{0}\bra{0}A\ket{0}\bra{0}+\ket{1}\bra{1}A\ket{1}\bra{1})\nonumber\\
	={}&\mathcal{E}^\mathrm{d}\Bigl(
	\Bigl[H,\ket{0}\bra{0}A\ket{0}\bra{0}\Bigr]
	+\Bigl[H,\ket{1}\bra{1}A\ket{1}\bra{1}\Bigr]
	\Bigr)\nonumber\\
	={}&\ket{0}\bra{0}\,\Bigl(
	\Bigl[H,\ket{0}\bra{0}A\ket{0}\bra{0}\Bigr]
	+\Bigl[H,\ket{1}\bra{1}A\ket{1}\bra{1}\Bigr]
	\Bigr)\,\ket{0}\bra{0}
	\nonumber\\
	&{}+\ket{1}\bra{1}\,\Bigl(
	\Bigl[H,\ket{0}\bra{0}A\ket{0}\bra{0}\Bigr]
	+\Bigl[H,\ket{1}\bra{1}A\ket{1}\bra{1}\Bigr]
	\Bigr)\,\ket{1}\bra{1}
	\nonumber\\
	={}&0,
\end{align}
for any $\rho\in\mathcal{T}(\mathbb{C}^2)$.
This shows that the Zeno Hamiltonian suppression with the dephasing channel $\mathcal{E}^\mathrm{d}$ works.
Since the dephasing channel $\mathcal{E}^\mathrm{d}$ itself is the peripheral projection of $\mathcal{E}^\mathrm{d}$, we have $(\mathcal{E}_\varphi^\mathrm{d})^n=\mathcal{E}^\mathrm{d}$ for its peripheral part.
The target Zeno evolution is thus given by $(\mathcal{E}_\varphi^\mathrm{d})^n\rme^{-\rmi t\mathcal{H}_\mathrm{Z}^{\mathcal{E}^\mathrm{d}}}=\mathcal{E}^\mathrm{d}$.
The convergence of $\mathcal{E}_\mathrm{Z}^\mathrm{d}$ to $\mathcal{E}^\mathrm{d}$ can be seen numerically by computing the distance $\Vert \Lambda(\mathcal{E}_\mathrm{Z}^\mathrm{d})-\Lambda(\mathcal{E}^\mathrm{d})\Vert_1$.
In Fig.~\ref{fig:Dephasing}(b), the worst-case distance $\max\Vert \Lambda(\mathcal{E}_\mathrm{Z}^\mathrm{d})-\Lambda(\mathcal{E}^\mathrm{d})\Vert_1$ for 100 randomly sampled Hamiltonians $H\in\mathcal{B}(\mathbb{C}^2)$ with $\Vert H\Vert_\infty=1$ is shown as a function of the number of Zeno steps $n$.
We infer that $\Vert \Lambda(\mathcal{E}_\mathrm{Z}^\mathrm{d})-\Lambda(\mathcal{E}^\mathrm{d})\Vert_1=\mathcal{O}(1/n)$ independently of the Hamiltonian $H$, and the Zeno Hamiltonian suppression indeed works.

More complicated non-ergodic channels without decoherence-free subsystems can be constructed using Lemma~\ref{lemma:no_DFS}, by taking the direct sum of ergodic channels.
Such channels have at least two cycles and work for the Zeno Hamiltonian suppression but do not work for the bath DD\@.

\subsection{Channel with Decoherence-Free Subsystem $\mathcal{E}^{\mathbb{1}/2}$}\label{sec:example_DFSS}
Let us next look at the quantum channel $\mathcal{E}^{\mathbb{1}/2}$ introduced in Example~\ref{example:DFSS}\eqref{it:DFSS1} with $\Omega=\frac{1}{2}\mathbb{1}$.
It is a two-qubit channel acting on $\mathcal{B}(\mathbb{C}_1^2\otimes\mathbb{C}_2^2)$.
It is defined by $\mathcal{E}^{\mathbb{1}/2}(A)=\tr_2(A)\otimes\frac{1}{2}\mathbb{1}$ for any $A\in\mathcal{B}(\mathbb{C}_1^2\otimes\mathbb{C}_2^2)$.
Here, $\Omega=\frac{1}{2}\mathbb{1}\in\mathcal{T}(\mathbb{C}_2^2)$ is the maximally mixed state.
The peripheral spectrum of $\mathcal{E}^{\mathbb{1}/2}$ is $\sigma_\varphi(\mathcal{E}^{\mathbb{1}/2})=\{1\}$, with the peripheral projection being $\mathcal{E}^{\mathbb{1}/2}$ itself.
However, its space of fixed points reads $\mathcal{F}(\mathcal{E}^{\mathbb{1}/2})=\{x\otimes\frac{1}{2}\mathbb{1}:x\in\mathcal{B}(\mathbb{C}_1^2)\}$, which is $4$-dimensional, so the channel $\mathcal{E}^{\mathbb{1}/2}$ is not ergodic. 
It admits a decoherence-free subsystem of $\mathcal{B}(\mathbb{C}_1^2)$.

Let us look at the Zeno Hamiltonian suppression with $\mathcal{E}^{\mathbb{1}/2}$.
To this end, take the Hamiltonian $H=Z\otimes\mathbb{1}\in\mathcal{B}(\mathbb{C}_1^2\otimes\mathbb{C}_2^2)$ and consider its adjoint representation $\mathcal{H}=[H,{}\bullet{}]$.
This Hamiltonian commutes with $\mathcal{E}^{\mathbb{1}/2}$, i.e.~$[\mathcal{H},\mathcal{E}^{\mathbb{1}/2}]=0$, and hence,
\begin{equation}
\mathcal{H}_\mathrm{Z}^{\mathcal{E}^{\mathbb{1}/2}}
=\mathcal{E}^{\mathbb{1}/2}\mathcal{H}\mathcal{E}^{\mathbb{1}/2}
=\mathcal{H}\mathcal{E}^{\mathbb{1}/2}.
\label{eq:DFSS_zeno_example}
\end{equation}
The Hamiltonian $\mathcal{H}$ is not cancelled, and it survives the Zeno projection.
Indeed,
\begin{align}
	\mathcal{E}_\mathrm{Z}^{\mathbb{1}/2}
	&=\Bigl(
	\mathcal{E}^{\mathbb{1}/2}\rme^{-\rmi\frac{t}{n}\mathcal{H}}\Bigr)^n\nonumber\\
	&=(\mathcal{E}_\varphi^{\mathbb{1}/2})^n\rme^{-\rmi t\mathcal{H}_\mathrm{Z}^{\mathcal{E}^{\mathbb{1}/2}}}\nonumber\\
	&=\mathcal{E}^{\mathbb{1}/2}\rme^{-\rmi t\mathcal{H}\mathcal{E}^{\mathbb{1}/2}}\nonumber\\
	&=\mathcal{E}^{\mathbb{1}/2}\rme^{-\rmi t\mathcal{H}}.
\end{align}
Here, we have used the fact that $\mathcal{E}^{\mathbb{1}/2}$ itself is the peripheral projection and thus $(\mathcal{E}_\varphi^{\mathbb{1}/2})^n=\mathcal{E}^{\mathbb{1}/2}$ for all $n\in\mathbb{N}$.
The Zeno Hamiltonian suppression does not work with $\mathcal{E}^{\mathbb{1}/2}$.
Likewise, the bath DD cannot work.

We confirm these numerically in Fig.~\ref{fig:DFSS}\@.
In Fig.~\ref{fig:DFSS}(a), the bath Hilbert space $\mathbb{C}_1^2\otimes\mathbb{C}_2^2$ is coupled with another single-qubit Hilbert space $\mathbb{C}_0^2$ via the Hamiltonian $H=Z\otimes Z\otimes\mathbb{1}\in\mathcal{B}(\mathbb{C}_0^2\otimes\mathbb{C}_1^2\otimes\mathbb{C}_2^2)$.
We compute the purity $\mathtt{P}(\Lambda_0^\mathrm{DD})$ of the reduced Choi-Jamio{\l}kowski state $\Lambda_0^\mathrm{DD}=\tr_{1,2}(\Lambda(\mathcal{E}_\mathrm{DD}^{\mathbb{1}/2}))$ for $\mathcal{E}_\mathrm{DD}^{\mathbb{1}/2}=\Bigl((\mathcal{I}_0\otimes\mathcal{E}^{\mathbb{1}/2})\rme^{-\rmi\frac{t}{n}\mathcal{H}}\Bigr)^n$.
The purity stays constant and does not approach $1$ as the number of decoupling steps $n$ increases.
The bath DD does not work.
Likewise, for the Zeno Hamiltonian suppression with the channel $\mathcal{E}^{\mathbb{1}/2}$, we numerically confirm that the Zeno evolution $\mathcal{E}_\mathrm{Z}^{\mathbb{1}/2}$ does not remove the Hamiltonian $\mathcal{H}=[H,{}\bullet{}]$ with $H=Z\otimes\mathbb{1}\in\mathcal{B}(\mathbb{C}_1^2\otimes\mathbb{C}_2^2)$.
This can be seen in Fig.~\ref{fig:DFSS}(b), where we compute the distance $\Vert \Lambda(\mathcal{E}_\mathrm{Z}^{\mathbb{1}/2})-\Lambda(\mathcal{E}^{\mathbb{1}/2})\Vert_1$ between the Choi-Jami{\l}kowski states of the Zeno evolution $\mathcal{E}_\mathrm{Z}^{\mathbb{1}/2}$ and of the target evolution $\mathcal{E}^{\mathbb{1}/2}$ without the Hamiltonian component.
This distance does not shrink with the number of Zeno steps $n$.
Hence, the Zeno Hamiltonian suppression does not work.
\begin{figure}
\begin{tabular}{rr}
\multicolumn{1}{l}{\footnotesize(a)}
&
\multicolumn{1}{l}{\footnotesize(b)}
\\
	\begin{tikzpicture}[mark size={1}, scale=1]
	\begin{axis}[
	xmin=1,
	xmax=100,
	ymin=0.25,
	ymax=1,
	ylabel near ticks,
	xlabel={\footnotesize Number of decoupling steps $n$},
	ylabel={\footnotesize System Choi state purity},
	x post scale=0.85,
	y post scale=0.7,
	xmode=log,
	tick label style={font=\footnotesize},
	legend style={font=\footnotesize},
	legend pos={south west},
	legend cell align=left,
	]
	\addplot[blue, only marks] table [x=n, y=P, col sep=comma] {DFSS_Purity_Bi.csv};
	\addlegendentry{$\mathtt{P}(\Lambda_0^\mathrm{DD})$ for $H=Z\otimes Z\otimes\mathbb{1}$};
	\addplot[red, only marks] table [x=n, y=P, col sep=comma] {DFSS_Purity_Bi_Rand.csv};
	\addlegendentry{$\mathbb{E}[\mathtt{P}(\Lambda_0^\mathrm{DD})]$ for random $H$};
	\addplot[color=red, domain=1:100,thick] {0.909};
	\end{axis}
	\end{tikzpicture}
&
  	\begin{tikzpicture}[mark size={1}, scale=1]
	\begin{axis}[
	xmin=1,
	xmax=100,
	ymin=0,
	ymax=2,
	ylabel near ticks,
	xlabel={\footnotesize Number of Zeno steps $n$},
	ylabel={\footnotesize Zeno error},
	x post scale=0.85,
	y post scale=0.7,
	xmode=log,
	tick label style={font=\footnotesize},
	legend style={font=\footnotesize, at={(0.98,0.75)},anchor=north east},
	legend cell align=left,
	]
	\addplot[blue,only marks] table [x=n, y=error, col sep=comma] {DFSS_Error_Mono.csv};
	\addplot[red,only marks] table [x=n, y=error, col sep=comma] {DFSS_Error_Mono_Rand.csv};
	\addplot[color=red, domain=1:100,thick] {0.553};
	\addlegendentry{$\Vert \Lambda(\mathcal{E}_\mathrm{Z}^{\mathbb{1}/2})-\Lambda(\mathcal{E}^{\mathbb{1}/2})\Vert_1$ for $H=Z\otimes\mathbb{1}$};
	\addlegendentry{$\mathbb{E}[\Vert \Lambda(\mathcal{E}_\mathrm{Z}^{\mathbb{1}/2})-\Lambda(\mathcal{E}^{\mathbb{1}/2})\Vert_1]$ for random $H$};
	\end{axis}
	\end{tikzpicture}
\end{tabular}
\caption{\label{fig:DFSS}Bath dynamical decoupling and Zeno Hamiltonian suppression with the two-qubit channel $\mathcal{E}^{\mathbb{1}/2}$ introduced in Example~\ref{example:DFSS}(i), which admits a decoherence-free subsystem on the first qubit. (a) The purity $\mathtt{P}(\Lambda_0^\mathrm{DD})$ of the reduced Choi-Jamio{\l}kowski state for the Hamiltonian $H=Z\otimes Z\otimes\mathbb{1}\in\mathcal{B}(\mathbb{C}_0^2\otimes\mathbb{C}_1^2\otimes\mathbb{C}_2^2)$ (blue) and the average purity $\mathbb{E}[\mathtt{P}(\Lambda_0^\mathrm{DD})]$ over 100 randomly sampled Hamiltonians $H\in\mathcal{B}(\mathbb{C}_0^2\otimes\mathbb{C}_1^2\otimes\mathbb{C}_2^2)$ normalized as $\|H\|_\infty=1$ (red), for the bath dynamical decoupling. The total evolution time is fixed at $t=1$. The former (blue) is constant at $\mathtt{P}(\Lambda_0^\mathrm{DD})\approx0.59$, while the latter (red) saturates to $\mathbb{E}[\mathtt{P}(\Lambda_0^\mathrm{DD})]\approx0.91$. The bath dynamical decoupling does not work. (b) The distance $\Vert \Lambda(\mathcal{E}_\mathrm{Z}^{\mathbb{1}/2})-\Lambda(\mathcal{E}^{\mathbb{1}/2})\Vert_1$ for the Hamiltonian $H=Z\otimes\mathbb{1}\in\mathcal{B}(\mathbb{C}_1^2\otimes\mathbb{C}_2^2)$ (blue) and the average distance $\mathbb{E}[\Vert \Lambda(\mathcal{E}_\mathrm{Z}^{\mathbb{1}/2})-\Lambda(\mathcal{E}^{\mathbb{1}/2})\Vert_1]$ over 100 randomly sampled Hamiltoanians $H\in\mathcal{B}(\mathbb{C}_1^2\otimes\mathbb{C}_2^2)$ normalized as $\|H\|_\infty=1$ (red), for the Zeno Hamiltonian suppression. The total evolution time is fixed at $t=1$. The former (blue) is constant at $\Vert \Lambda(\mathcal{E}_\mathrm{Z}^{\mathbb{1}/2})-\Lambda(\mathcal{E}^{\mathbb{1}/2})\Vert_1\approx1.68$, while the latter (red) approaches $\mathbb{E}[\Vert \Lambda(\mathcal{E}_\mathrm{Z}^{\mathbb{1}/2})-\Lambda(\mathcal{E}^{\mathbb{1}/2})\Vert_1]\approx0.55$. The Zeno Hamiltonian suppression does not work.}
\end{figure}
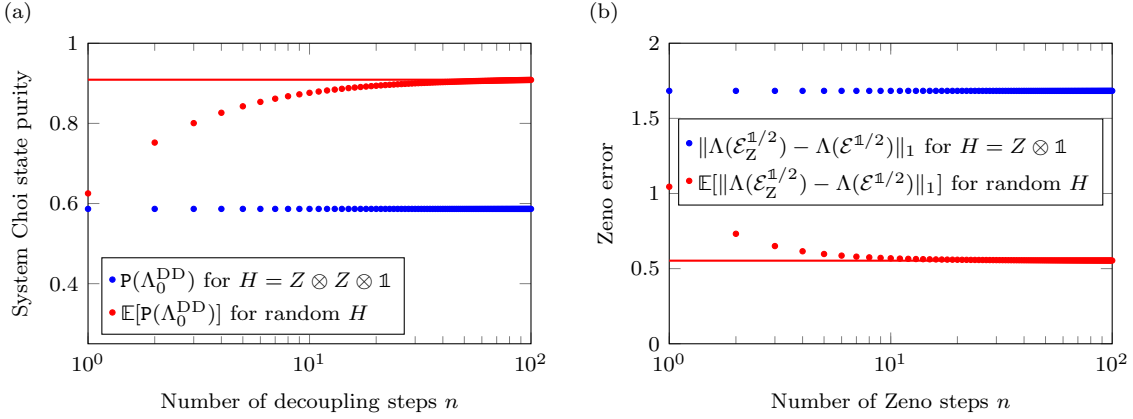

\section{Conclusion}\label{sec:conclusion}
In this paper, we generalized the notion of dynamical decoupling from cycles of unitary kicks to repeated CPTP kicks.
This approach is physically motivated by the application of the DD sequence to the bath.
In this case, it does not matter that the decoupling operations destroy the coherence of the input state as we only care about the \emph{system} state at the end of the evolution.
We found that this procedure of bath DD works if and only if the applied quantum channel is ergodic.
Thus, our method provides a true generalization of the standard unitary DD, where the effective channel after each cycle is mixing.
Furthermore, we study under which condition the repeated application of a quantum channel completely suppresses the Hamiltonian of a mono-partite system \`a la quantum Zeno effect.
It turns out that a weaker condition than ergodicity suffices for this to happen.
In particular, the quantum channel must not admit any decoherence-free subsystems to induce the quantum Zeno dynamics without any Hamiltonian component.
As in the case of bath DD, this condition is both necessary and sufficient.
To arrive at these results, we proved some spectral properties of quantum channels, which might be of independent interest beyond the scope of this paper.
In particular, we characterized the peripheral projections of ergodic quantum channels.
Furthermore, we characterized all quantum channels that are neither ergodic nor have decoherence-free subsystems in terms of the action on their respective peripheral (recurrent) subspace.

We believe that our results have applications in quantum technologies when the coherence times of quantum systems should be enhanced.
In this regard, solid-state systems might be particularly well suited for bath DD as it may be feasible to physically address and control impurities that act as a bath.
Importantly, the bath DD scheme can be combined with the standard unitary system DD to further improve the decoupling fidelity.
Physically, the bath DD might be achieved either through coupling the bath to a larger quantum system or by applying noisy operations to the bath.

Finally, we would like to remark on two possibilities for further generalizations of our results.
\begin{enumerate}
	\item It would be interesting to derive bath dynamical decoupling conditions for \emph{cycles} of quantum channels.
		In this work, we only considered repeated applications of a single quantum channel.
		However, cycles of different quantum channels might potentially allow to ease the decoupling conditions even further.
	\item Another relevant generalization of our results would be to study infinite-dimensional baths, such as bosonic environments.
		Such models include thermal noise so that one could potentially find conditions, under which bath heating leads to decoupling.
\end{enumerate}

\section*{Acknowledgments}
We acknowledge several interesting discussions about experimental realizations of bath dynamical decoupling with Jemy Geordy, Sarath Raman Nair, and Thomas Volz.
Furthermore, we thank Daniele Amato and Arturo Konderak for their valuable feedback on the manuscript.

DB acknowledges funding from the Australian Research Council (project numbers FT190100106, DP210101367, CE170100009).
AH was partially supported by the Sydney Quantum Academy.
KY acknowledges support from the Top Global University Project from the Ministry of Education, Culture, Sports, Science and Technology (MEXT), Japan, and the supports by JSPS KAKENHI Grant Numbers JP18K03470, JP18KK0073, and JP24K06904, from the Japan Society for the Promotion of Science (JSPS).

\bibliographystyle{prsty-title-hyperref}
\bibliography{20241126_Bath_DD.bib}
\end{document}